\documentclass[preprint]{elsarticle}
\pdfoutput=1
\usepackage{./0}
\usepackage{appendix}

\begin{document}

% frontmatter
\begin{frontmatter}
\title{Descriptive complexity for pictures languages \\ (Extended abstract)}
\author[caen]{Etienne Grandjean}\ead{etienne.grandjean@unicaen.fr}
\author[marseille]{Fr\'ed\'eric Olive}\ead{frederic.olive@lif.univ-mrs.fr}
\author[caen]{Ga\'etan Richard}\ead{gaetan.richard@info.unicaen.fr}
\address[caen]{Universit\'e de Caen / ENSICAEN / CNRS - GREYC - Caen, France}
\address[marseille]{Aix-Marseille Universit\'e - LIF - Marseille, France}
%
% abstract
\begin{abstract}
This paper deals with descriptive complexity of picture languages of any dimension by syntactical fragments of existential second-order logic.
\begin{itemize}
\item  We uniformly generalize to any dimension the characterization by Giammarresi et al.~\cite{GRST96} of the class of \emph{recognizable} picture languages in existential monadic second-order logic.
\item  We state several logical characterizations of the class of picture languages recognized in linear time on nondeterministic cellular automata of any dimension. They are the first machine-independent characterizations of complexity classes of cellular automata.
\end{itemize}
Our characterizations are essentially deduced from normalization results we prove for first-order and existential second-order logics over pictures. They are obtained in a general and uniform framework that allows to extend them to other "regular" structures. Finally, we describe some hierarchy results that show the optimality of our logical characterizations and delineate their limits.
\end{abstract}
% keywords
\begin{keyword} Picture languages \sep  locality \sep  recognizability \sep  linear time \sep  cellular automata \sep  logical characterizations \sep  existential second-order logic. \end{keyword}
\end{frontmatter}%

% section : INTRO
\section*{Introduction}\label{sec:introduction}
One goal of descriptive complexity is to establish logical characterizations of natural classes of problems in finite model theory. Many results in this area involve second-order logic ($\SO{}$) and its restrictions, monadic second-order logic ($\mso$) and existential second-order logic ($\eso{}$).
Indeed, there are two lines of research that roughly correspond to either of these restrictions:

\paragraph{\textit{(a)} The formal language current} It starts from the pioneering result by B\"uchi, Elgot and Trahtenbrot~\cite{Buchi60}, which states that the class of regular languages equals the class of $\mso$-definable languages; in short, $\reg{}=\mso$. This line of research aims at characterizing in logic the natural classes of algebraically defined languages (sets of words) or sets of structures (trees, graphs, \etc.) defined by finite state recognizability or local properties such as tilings.
\paragraph{\textit{(b)} The computational complexity current} It originates from another famous result, Fagin's Theorem \cite{Fagin74}, which characterizes the class $\np$ as the class of problems definable in $\eso{}$.

\medskip

For many years, both directions of research have produced plenty of results: see \eg \cite{EbbinghausF95,Libkin04} for descriptive complexity of formal languages and~\cite{EbbinghausF95,Immerman99,GradelKLMSVVW07,Libkin04} for the one of complexity classes. However, and this may be surprising, only few connections are known between those two areas of descriptive complexity. Of course, an explanation is that formal language theory has its own purposes that have little to do with complexity theory. In our opinion, the main reason is that while $\mso$ logic \emph{exactly fits} the fundamental notion of \emph{recognizability}, as exemplified in the work of Courcelle~\cite{CourcelleE12}, this logic seems \emph{transversal} to computational complexity. We argue this is due to the intrinsic \emph{locality} that $\mso$ logic \emph{inherits} from first-order logic~\cite{Hanf65,Gaifman82}. Typically, whereas $\mso$, or even existential $\mso$ ($\emso$), expresses some $\np$-complete graph problems such as 3-colourability, it cannot express some other ones such as Hamiltonicity (see~\cite{Turan84,deRougemont87,Libkin04}) or even some tractable graph properties, such as the existence of a perfect matching.
In contrast, the situation is very clear on trees as on words: $\mso$ only captures the class of "easiest" problems; an extension of B\"uchi's Theorem~\cite{ThatcherW68} states that a tree language is $\mso$ definable iff it is recognizable by a finite tree automaton

\medskip

Thus, items \textit{(a)} and \textit{(b)} above %, the formal language current of descriptive complexity involving $\mso$, and item $(b)$, the computational complexity current that involves $\eso{}$, 
seem quite separate for problems on words, trees or graphs, in one case (for words and trees) because MSO only expresses easy problems (regular languages), in the other case because MSO and EMSO do \emph{not} correspond to any complexity class over graphs.

\medskip

What about \emph{picture languages}, that is sets of $d$-pictures, \ie, $d$-dimensional words (or coloured grids)? First, notice the following results:
\begin{enumerate}
\item  In a series of papers culminating in~\cite{GRST96}, Giammarresi and al. have proved that a $2$-picture language is \emph{recognizable}, \ie is the projection of a local %(that means: \emph{tilable}) 
$2$-picture language, iff it is definable in $\emso$. In short: $\rec{2}=\emso$. 
\item  In fact, the class $\rec{2}$ contains some $\np$-complete problems. More generally, one observes that for each dimension $d\ge1$, $\rec{d}$ \emph{can be defined as} the class of $d$-picture languages recognized by nondeterministic $d$-dimensional \emph{cellular automata} in \emph{constant time}\footnote{That means: for such a picture language $L$, there is some constant integer $c$ such that each computation stops at instant $c$, and  $p\in L$ iff it has at least one computation whose final configuration is accepting in the following sense: \emph{all the cells are in an accepting state.}}.
\end{enumerate}

In some sense, the present paper is an attempt to bridge the gap between the formal language current of descriptive complexity involving $\mso$ and the computational complexity current that involves $\eso{}$. This paper originates from two questions about word/picture languages:
\begin{enumerate}
\item How can we generalize the proof of the above-mentioned theorem of Giammarresi and al. to any dimension? That is, can we establish the equality $\rec{d}=\emso$ for any $d\ge 1$?
\item Can we obtain logical characterizations of time complexity classes of cellular automata\footnote{This originates from a question that J. Mazoyer asked us in 2000 (personal communication): give a logical characterization of the linear time complexity class of nondeterministic cellular automata.}?
\end{enumerate}
The paper addresses both questions; it also compares, in a common framework, the point of view of formal language theory with that of computational complexity. 
A $d$-picture language over an alphabet $\Sigma$ is a set of $d$-pictures $p:[1,n]^d\to\Sigma$, \ie, $d$-dimensional $\Sigma$-words%
\footnote{More generally, the domain of a $d$-picture is of the "rectangular" form $[1,n_1]\times\ldots\times[1,n_d]$. 
For simplicity and uniformity of presentation, we have chosen to present the results of this paper in the particular case of "square" pictures of domain $[n]^d$. Fortunately, our results also hold with the same proofs for general domains $[1,n_1]\times\ldots\times[1,n_d]$.}. There are two natural manners to represent a $d$-picture $p$ as a first-order structure:
\begin{itemize}
\item as a \emph{pixel structure:} on the \emph{pixel} domain $[1,n]^d$;
\item as a \emph{coordinate structure:} on the \emph{coordinate} domain $[1,n]$.
\end{itemize}

Significantly, these two representations respectively correspond to the two above-mentioned points of view as shown by our results.

\subsection*{Our results:}
We establish two kinds of logical characterizations of $d$-picture languages, for all dimensions $d\ge1$:
\begin{enumerate}
\item  \emph{On pixel structures:} $\rec{d}=\eso{}(\arity 1)=\eso{}(\var 1)=\eso{}(\fa^1,\arity 1)$. That means a $d$-picture language is \emph{recognizable} iff it is definable in monadic $\eso{}$ (resp. in $\eso{}$ with $1$ first-order variable, or in monadic $\eso{}$ with $1$ universally quantified first-order variable). 
\item  \emph{On coordinate structures:} $\nlinca^{d}=\eso{}(\var d+1)=\eso{}(\fa^{d+1},\arity d+1)$; that means a $d$-picture language is recognized by a nondeterministic $d$-dimensional cellular automaton in \emph{linear time} iff it is definable in $\eso{}$ with $d+1$ distinct first-order variables (resp. $\eso{}$ with second-order variables of arity at most $d+1$ and a prenex first-order part of prefix 
$\forall^{d+1}$).  
\end {enumerate}
Both items (1) and (2) are easy consequences of \emph{normalization results} of, respectively, first-order and $\eso{}$ logics we prove over picture languages. In particular, the "normalization" equality $\eso{}(\arity 1)=\eso{}(\fa^1,\arity 1)$ is a consequence of the fact that on pixel structures (and more generally, on structures that consist of bijective functions and unary relations) any first-order formula is equivalent to a boolean combination of "cardinality" formulas of the form: "there exists $k$ distinct elements $x$ such that $\psi(x)$", where $\psi$ is a quantifier-free formula with only \emph{one} variable. The "normalization" equality explicitly expresses the "local" feature of $\mso$ on pictures and can be generalized to other "regular" structures.

Thirdly, in contrast with (1) and (2), we establish several \emph{strict hierarchy} results for any fixed $d\ge 2$ and for $d$-picture languages represented by coordinate structures: in particular, we prove $$\eso{}(\var d-1)\subsetneq \rec{d}\subsetneq \eso{}(\var d)=\nlinca^{d}$$ and $\eso{}(\var d)\subsetneq \eso{}(\arity d)$. Here, the fact that \emph{no} natural restriction of $\eso{}$ (in particular, $\eso{}(\var d)$) exactly captures the class $\rec{d}$ for the coordinate representation of pictures seems to us the symptom of the large expressiveness of this logic: significantly, the \emph{strict inclusion} $\eso{}(\var d)\subsetneq \eso{}(\arity d)$ for $d$-pictures in the coordinate representation, if $d\geq 2$, strongly contrasts with the \emph{equality} $\eso{}(\arity 1)=\eso{}(\var 1)$ that holds in the pixel representation.

\medskip

In this document, we sketch many proofs and omit other ones, in particular, the proofs of hierarchy theorems, in Section~\ref{sec:hierarchy}. However, the very technical proof of the normalization of the logic $\eso{}(\fa d, \arity d)$ on coordinate encodings of $(d-1)$-pictures (for $d\ge 2$) is completely described in the appendix.

% SECTION % SECTION % SECTION % SECTION % SECTION % SECTION % SECTION % SECTION % SECTION % SECTION % SECTION % SECTION % SECTION % SECTION % SECTION % SECTION % SECTION % SECTION % SECTION % SECTION % SECTION % SECTION 
% SECTION % SECTION % SECTION % SECTION % SECTION % SECTION % SECTION % SECTION % SECTION % SECTION % SECTION % SECTION % SECTION % SECTION % SECTION % SECTION % SECTION % SECTION % SECTION % SECTION % SECTION % SECTION 
\section{Preliminaries}\label{sec:preliminaries}
In the definitions below and all along the paper, we denote by $\Sigma$, $\Gamma$ some finite alphabets and by $d$ a positive integer. For any positive integer $n$, we set $[n]\egaldef\{1,\dots,n\}$. We are interested in sets of pictures of any fixed dimension $d$.

% definition
\begin{definition}\label{def:picture}
A \textdef{$d$-dimensional picture} or \textdef{$d$-picture} on $\Sigma$ is a function $\apic: [n]^d\to\Sigma$ where $n$ is a positive integer. The set \textdef{$dom(\apic)$}$=[n]^d$ is called the \textdef{domain} of picture $\apic$ and its elements are called \textdef{points}, \textdef{pixels} or \textdef{cells} of the picture. A set of $d$-pictures on $\Sigma$ is called a \textdef{$d$-dimensional language}, or \textdef{$d$-language}, on $\Sigma$.
\end{definition}

Notice that $1$-pictures on $\Sigma$ are nothing but nonempty words on $\Sigma$. 

% subsection
\subsection{Pictures as model theoretic structures}\label{sec:logiques}

Along the paper, we will often describe $d$-languages as sets of models of logical formulas. To allow this point of view, we must settle on an encoding of $d$-pictures as model theoretic structures. 

\medskip

For logical aspects of this paper, we refer to the usual definitions and notations in logic and finite model theory (see~\cite{EbbinghausF95} or~\cite{Libkin04}, for instance). A \textdef{signature} (or \textdef{vocabulary}) $\sigma$ is a finite set of relation and function symbols each of which has a fixed arity. A (finite) \textdef{structure} $S$ of vocabulary $\sigma$, or $\sigma$-structure, consists of a finite domain $D$ of cardinality $n\ge 1$, and, for any symbol $s\in\sigma$, an interpretation of $s$ over $D$, often denoted by $s$ for simplicity.  The tuple of the interpretations of the $\s$-symbols over $D$ is called the \textdef{interpretation} of $\s$ over $D$ and, when no confusion results, it is also denoted $\s$. The \textdef{cardinality of a structure} is the cardinality of its domain. For any signature $\s$, we denote by \textdef{$\struc({\s})$} the class of (finite) $\s$-structures. We write \textdef{$\model(\Phi)$} the set of $\sigma$-structures which satisfy some fixed formula $\Phi$. We will often deal with \textdef{tuples} of objects. We denote them by bold letters. 

\medskip

There are two natural manners to represent a picture by some logical structure: on the domain of its pixels, or on the domain of its coordinates. This gives rise to the following definitions:

% definition
\begin{definition}\label{def:pixel structure}
Given $\apic:\dom n^{d}\imp\Sigma$, we denote by \textdef{$\picode{d}(\apic)$} the structure \[\picode{d}(\apic)=  ([n]^{d}, (Q_s)_{s\in\Sigma},(\suc_i)_{i\in [d ]}, (\min_i)_{i\in [d ]}, (\max_i)_{i\in [d ]}).\] 
Here: 
\begin{itemize}
\item  $\suc_j$ is the \textdef{(cyclic) successor function} according to the $j^{th}$ dimension of $[n]^{d}$, that is: for each $a=(a_i)_{i \in [d ]} \in [n]^{d}$, $\suc_j(a) = a^{(j)}$ where $a^{(j)}_i$, the $i^{th}$ component of $a^{(j)}$, equals $a_i$, the $i^{th}$ component $a_i$ of $a$, except the $j^{th}$ one which equals the cyclic successor of the $j^{th}$ component $a_j$ of $a$. More formally: 
\begin{itemize}
\item  $a^{(j)}_j = a_j+1$ if $a_j<n$, and $a^{(j)}_j=1$ otherwise;
\item  $a^{(j)}_i = a_i$ for each $i \ne j$.
\end{itemize}   
\item  the $\min_i$'s, $\max_i$'s and $Q_s$'s are the following unary (monadic) relations:
\begin{itemize}
\item  $\min_i =  \{ a \in [n]^{d}:  a_i=1 \} $;
\item  $\max_i = \{ a \in [n]^{d}: a_i=n_i \} $;
\item  $Q_s = \{ a \in  [n]^{d}: \apic(a)=s  \}$.
\end{itemize}
\end{itemize} 
\end{definition}

% definition
\begin{definition}\label{def:coord structure}
Given $\apic:\dom n^{d}\imp\Sigma$, we denote by \textdef{$\cocode{d}(\apic)$} the structure 
\begin{equation}\label{eq:def cocode}
\cocode{d}(\apic)=\la\dom n,(Q_s)_{s\in\Sigma},<,\suc,\min,\max\ra.
\end{equation}
Here:
\begin{itemize}
\item  Each $Q_s$ is a $d$-ary relation symbol interpreted as the set of cells of $\apic$ labelled by an $s$. In other words: $Q_s=\{a\in\dom n^d : p(a)=s\}$. 
\item  $<$, $\min$, $\max$ are predefined relation symbols of respective arities $2$, $1$, $1$, that are interpreted, respectively, as the sets $\{(i,j): 1\leq i<j\leq n\}$, $\{1\}$ and $\{n\}$. 
\item  $\suc$ is a unary function symbol interpreted as the cyclic successor. (That is: $\suc(i)=i+1$ for $i<n$ and $\suc(n)=1$.) 
\end{itemize}
\end{definition}
For a $d$-language $L$, we set $\textdef{$\picode{d}(L)$}=\{\picode{d}(\apic) : \apic \in L\}$ and $\textdef{$\cocode{d}(L)$}=\{\cocode{d}(\apic):\apic\in L\}$.

% subsection 
\subsection{Logics under consideration}\label{sec:logiques}

Let us now come to the logics involved in the paper. 
All formulas considered hereafter belong to \textdef{relational Existential Second-Order logic}. Given a signature $\sg$, indifferently made of relational and functional symbols, a relational existential second-order formula of signature $\sg$ has the shape $\Phi\equiv\ex\tu R\phi(\sg,\tu R)$, where $\tu R=(R_1,\dots,R_k)$ is a tuple of relational symbols and $\phi$ is a first-order formula of signature $\sg\cup\{\tu R\}$. We denote by \textdef{$\eso{\sg}$} the class thus defined. We will often omit to mention $\sg$ for considerations on these logics that do not depend on the signature. Hence, \textdef{$\eso{}$} stands for the class of all formulas belonging to $\eso{\sg}$ for some $\sg$. 

\medskip

We will pay great attention to several variants of $\esorel{}$. In particular, we will distinguish formulas of type $\Phi\equiv\ex\tu R\phi(\sg,\tu R)$ according to: 
\begin{itemize}
\item[-]  the number of distinct first-order variables involved in $\phi$, 
\item[-]  the arity of the second-order symbols $R\in\tu R$, and 
\item[-]  the quantifier prefix of some prenex form of $\phi$. 
\end{itemize}

With the logic \textdef{$\esorelsig(\fa^d,\arity \ell)$}, we control these three parameters: it is made of formulas of which first-order part is prenex with a universal quantifier prefix of length $d$, and where existentially quantified relation symbols are of arity at most $\ell$. In other words, $\esorelsig(\fa^d,\arity \ell)$ collects formulas of shape: $$\exists\tu R\fa\tu x\theta(\s,\tu R,\tu x)$$ where $\theta$ is quantifier free, $\tu x$ is a $d$-tuple of first-order variables, and $\tu R$ is a tuple of relation symbols of arity smaller than $\ell$. Relaxing some constraints of the above definition, we set:
\begin{eqnarray*}
\textdef{$\esorelsig(\fa^d)$}       =  & \dst{\bigcup_{\ell>0}}\esorelsig(\fa^{d},\arity \ell)\AND %\\ 
\textdef{$\esorelsig(\arity \ell)$} = & \dst{\bigcup_{d>0}}\esorelsig(\fa^{d},\arity \ell).
\end{eqnarray*}

Finally, we write \textdef{$\esorelsig(\var d)$} for the class of formulas that involve at most $d$ first-order variables, thus focusing on the sole number of distinct first-order variables (possibly quantified several times). 

% SECTION % SECTION % SECTION % SECTION % SECTION % SECTION % SECTION % SECTION % SECTION % SECTION % SECTION % SECTION % SECTION % SECTION % SECTION % SECTION % SECTION % SECTION % SECTION % SECTION % SECTION % SECTION 
% SECTION % SECTION % SECTION % SECTION % SECTION % SECTION % SECTION % SECTION % SECTION % SECTION % SECTION % SECTION % SECTION % SECTION % SECTION % SECTION % SECTION % SECTION % SECTION % SECTION % SECTION % SECTION 
% section : CARAC REC
\section{A logical characterization of $\rec{}$}\label{sec:eso1 to rec}
In order to define a notion of locality based on sub-pictures we need to mark the border of each picture.

% definition
\begin{definition}\titledef{bordered picture}
By \textdef{$\Gamma^{\sharp}$} we denote the alphabet $\Gamma\cup\{\sharp\}$ where $\sharp$ is a special symbol not in $\Gamma$. Let $\apic$ be any $d$-picture of domain  $[n]^d$ on $\Gamma$. The \textdef{bordered $d$-picture} of $\apic$, denoted by \textdef{$\apic^\sharp$}, is the $d$-picture $\apic^\sharp : [0,n+1]^d\to\Gamma^{\sharp}$ defined by $$\apic^\sharp(a)=\left\{\begin{array}{l} \apic(a)\IF a\in dom(\apic) ; \\ \sharp\OTHERWISE.  \end{array}\right.$$ Here, "otherwise" means that $a$ is on the border of $\apic^\sharp$, that is, some component $a_i$ of $a$ is $0$ or $n+1$.
\end{definition}

Let us now define our notion of \emph{local picture language}. It is based on some sets of allowed patterns (called tiles) of the bordered pictures.

% definition
\begin{definition}\label{def:local}\titledef{Local language and tiling} 
\begin{enumerate}
\item  Given a $d$-picture $\apic$ and an integer $j\in[d]$, two cells $a=(a_i)_{i\in [d]}$ and $b=(b_i)_{i\in [d ]}$ of $\apic$ are \textdef{$j$-adjacent} if they have the same coordinates, except the $j^{th}$ one for which $|a_j-b_j|=1$. 
\item A \textdef{tile} for a $d$-language $L$ on $\Gamma$ is a couple in $(\Gamma^\sharp)^2$. 
\item  A picture \textdef{$\apic$ is $j$-tiled} by a set of tiles $\Delta\subseteq(\Gamma^\sharp)^2$ if for any two $j$-adjacent points $a,b\in dom(\apic^\sharp)$: $(\apic^\sharp(a),\apic^\sharp(b))\in\Delta.$
\item  Given $d$ sets of tiles $\Delta_1,\dots,\Delta_d\subseteq (\Gamma^\sharp)^2$, a $d$-picture \textdef{$\apic$ is tiled by $(\Delta_1$,\dots,$\Delta_d)$} if $\apic$ is $j$-tiled by $\Delta_{j}$ for each $j\in[d]$.
\item  We denote by \textdef{$L(\Delta_1$,\dots,$\Delta_d)$} the set of $d$-pictures on $\Gamma$ that are tiled by $(\Delta_1$,\dots,$\Delta_d)$. 
\item  A $d$-language $L$ on $\Gamma$ is \textdef{local} if there exist $\Delta_1,\dots,\Delta_d\subseteq (\Gamma^\sharp)^2$ such that $L=L(\Delta_1$,\dots,$\Delta_d)$. We then say that $L$ is \textdef{$(\Delta_1,\dots,\Delta_d)$-local}, or \textdef{$(\Delta_1,\dots,\Delta_d)$-tiled}. 
\end{enumerate}
\end{definition}

% definition
\begin{definition}\label{def:rec}\titledef{Recognizable language}
A $d$-language $L$ on $\Sigma$ is \textdef{recognizable} if it is the projection of a \emph{local} $d$-language over an alphabet $\Gamma$. It means there exist a surjective function $\pi:\Gamma\to\Sigma$ and a local $d$-language $L_{\local}$ on $\Gamma$ such that  $L=\{\pi\circ\apic:\ \apic\in L_{\local}\}.$ 

Because of Definition~\ref{def:local}, it also means there exist a surjective function $\pi:\Gamma\to\Sigma$ and $d$ subsets $\Delta_1,\dots,\Delta_d$ of $(\Gamma^\sharp)^2$ such that $$L=\{\pi\circ\apic:\ \apic\in L(\Delta_1,\dots,\Delta_d)\}.$$ 
We write \textdef{$\rec{d}$} for the class of recognizable $d$-languages.
\end{definition}

% remark
\begin{remark}
Our notion of \emph{locality} is weaker than the one given by Giammarresi and al.~\cite{GRST96}. But this doesn't affect the meaning of \emph{recognizability}, which coincides with that used in~\cite{GRST96}. This confirms the robustness of this latter notion. 
\end{remark}

A characterization of recognizable languages of dimension 2 by existential monadic second-order logic was proved by by Giammarresi et al.~\cite{GRST96}. They established: 
% theorem
\begin{theorem}[\cite{GRST96}]\label{thm:Giammarresi & al}
For any $2$-dimensional language $L$,  $L\in\rec{2}\Ssi\picode{2}(L)\in\eso{}(\arity 1).$
\end{theorem} 
In this section, we come back to this result. We simplify its proof, refine the logic it involves, and generalize its scope to any dimension.

% theorem
\begin{theorem}\label{thm:rec ssi eso(fa 1) ssi eso(fa 1,ar 1)}
For any $d>0$ and any $d$-language $L$, the following assertions are equivalent:
\begin{enumerate}
\item\label{item:rec}  $L\in\rec{d}$;
\item\label{item:fa 1, ar 1}  $\picode{d}(L)\in\eso{}(\fa^1,\arity 1)$;
\item\label{item:ar 1}  $\picode{d}(L)\in\eso{}(\arity 1)$.
\end{enumerate}
\end{theorem}

Proposition~\ref{thm:rec <-> pixel in eso(1)} states the equivalence $\ref{item:rec}\Ssi\ref{item:fa 1, ar 1}$. In Proposition~\ref{thm:eso(1)=eso(1,1) on pixel}, we establish the normalization $\eso{}(\arity 1)=\eso{}(\fa^1,\arity 1)$ on pixel structures, from which the equivalence~$\ref{item:fa 1, ar 1}\Ssi\ref{item:ar 1}$ immediately follows. 

% proposition
\begin{proposition}\label{thm:rec <-> pixel in eso(1)}
For any $d>0$ and any $d$-language $L$ on $\Sigma$: \ $L\in\rec{d}\Ssi\picode{d}(L)\in\eso{}(\fa^1,\arity 1).$  
\end{proposition}

% proofsketch
\begin{proofsketch}
% condition necessaire
\framebox{$\Imp$} 
A picture belongs to $L$ if there exists a tiling of its domain whose projection coincides with its content. In the logic involved in the proposition, the ``$\arity 1$'' corresponds to formulating the existence of the tiling, while the $\fa^1$ is the syntactic resource needed to express that the tiling behaves as expected. Let us detail these considerations. 

\medskip

By Definition~\ref{def:local}, there exist an alphabet $\Gamma$ (which can be assumed disjoint from $\Sigma$), a surjective function $\pi:\Gamma\to\Sigma$ and $d$ subsets $\Delta_1,\dots,\Delta_d\subseteq (\Gamma^\sharp)^2$ such that $L$ is the set $\{\pi\circ\apic' : \apic'\in L(\Delta_1,\dots,\Delta_d)\}.$ 

The belonging of a picture $\apic': [n]^d\to\Gamma$ to $L(\Delta_1,\dots,\Delta_d)$ is easily expressed on $\picode{d}(\apic')=\la\dom n^d,(Q_s)_{s\in\Gamma},\dots\ra$ with a first-order formula which asserts, for each dimension $i\in\dom d$, that for any pixel $x$ of $\apic'$, the couple ($x,\suc_i(x))$ can be tiled with some element of $\Delta_i$. Because it deals with each cell $x$ separately, this formula has the form $\fa x\Psi(x,(Q_s)_{s\in\Gamma})$, where $\Psi$ is quantifier-free. 

Now, a picture $\apic: [n]^d\to\Sigma$ belongs to $L$ iff it results from a $\pi$-renaming of a picture $\apic'\in L(\Delta_1,\dots,\Delta_d)$. It means there exists a $\Gamma$-labeling of $\apic$ (that is, a tuple $(Q_s)_{s\in\Gamma}$ of subsets of $[n]^{d}$) corresponding to a picture of $L(\Delta_1,\dots,\Delta_d)$ (\ie fulfilling $\fa x\Psi(x,(Q_s)_{s\in\Gamma})$) and from which the actual $\Sigma$-labeling of $\apic$ (that is, the subsets $(Q_s)_{s\in\Sigma}$) is obtained \via $\pi$ (easily expressed by a formula of the form $\fa x\Psi'(x,(Q_s)_{s\in\Sigma},(Q_s)_{s\in\Gamma})$). 

Finally, the formula $(\ex Q_s)_{s\in\Gamma}\fa x : \Psi\et\Psi'$ conveys the desired property and fits the required form. 

\medskip

% condition suffisante
\framebox{$\Rimp$} 
In order to prove the converse implication, it is convenient to first normalize the sentences of $\eso{}(\fa^1,\arity 1)$. This is the role of the technical result below, which asserts that on pixel encodings, each such sentence can be rewritten in a very local form where the first-order part alludes only pairs of adjacent pixels of the bordered picture. We state it without proof: 

% fact
\begin{fact}\label{thm:localisation}
On pixel structures, any $\phi\in\eso{}(\fa^1,\arity 1)$ is equivalent to a sentence of the form: 
\begin{equation}\label{eq:localisation}
\ex\bold{U}\fa x\Et_{i\in [d]}\left\{ 
\begin{array}{rcll}
\min_i(x) 		& \imp & \textsf{m}_i(x) & \et \\
\max_i(x) 		& \imp & \textsf{M}_i(x) & \et \\
\neg \max_i(x) & \imp & \Psi_i(x)
\end{array}\right\}.
\end{equation}
Here, $\tu U$ is a list of monadic relation variables and $\textsf{m}_i$, $\textsf{M}_i$, $\Psi_i$ are quantifier-free formulas such that
\begin{itemize}
\item  atoms of $\textsf{m}_i$ and $\textsf{M}_i$ have all the form $Q(x)$;
\item  atoms of $\Psi_i$ have all the form $Q(x)$ or $Q(\suc_i(x))$,
\end{itemize}
where, in both cases, $Q\in\{ (Q_s)_{s\in\Sigma},\bold{U}\}$.%
\end{fact}

\medskip

Now, consider $L$ such that $\picode{d}(L)\in\eso{}(\fa^1,\arity 1)$. Fact~\ref{thm:localisation} ensures that $\picode{d}(L)$ is characterized by a sentence of the form~\refeq{eq:localisation} above. We have to prove that $L$ is the projection of some local $d$-language $L_{\local}$ on some alphabet $\Gamma$, that is a ($\Delta_1$,\dots, $\Delta_d$)-tiled language for some  $\Delta_1$,\dots, $\Delta_d\subseteq\Gamma^2$. Let $U_1,\dots, U_k$ denote the list of (distinct) elements of the set $\{(Q_s)_{s\in\Sigma},\bold{U}\}$ of unary relation symbols of $\phi$ so that the first ones $U_1,\dots, U_m$ are the $Q_s$'s (here, $\min_i$ and $\max_i$ symbols are excluded). The trick is to put each subformula $\textsf{m}_i(x)$, $\textsf{M}_i(x)$ and $\Psi_i(x)$ of $\phi$ into its \emph{complete disjunctive normal form} with respect to $U_1,\dots, U_k$. Typically, each subformula $\Psi_i(x)$ whose atoms are of the form 
$U_j(x)$  or  $U_j(\suc_i(x))$, for some $j\in [k]$, is transformed into the following "complete disjunctive normal form":
%$$
\begin{equation}\label{eq:complete disjunction}
\Ou_{(\mathbf{\eps}, \mathbf{\eps'})\in\Delta_i} 
\left( 
\Et_{j\in [k]}\eps_j U_j(x)\et\Et_{j\in [k]}\eps'_j U_j(\suc_i(x))
\right).
\end{equation}
%$$
Here, the following conventions are adopted:
\begin{itemize}
\item  $\mathbf{\eps}=(\eps_1,\dots, \eps_k)\in \{ 0,1 \}^k$ and similarly for $\mathbf{\eps'}$;
\item  for any atom $\alpha$ and any bit $\eps_{j}\in\{0,1\}$, $\eps_{j}\alpha$ denotes the literal $\alpha$ if $\eps_{j}=1$, the literal $\neg\alpha$ otherwise. 
\end{itemize}
For $\mathbf{\eps}\in\{ 0,1 \}^k$, we denote by $\Theta_{\mathbf{\eps}}(x)$ the "complete conjunction" $\Et_{j\in [k]} \eps_j U_j(x)$. Intuitively, $\Theta_{\mathbf{\eps}}(x)$ is a complete description of $x$ and the set $\Gamma=\bigcup_{i\in [m]} \{ 0^{i-1}10^{m-i} \} \times \{ 0,1 \}^{k-m}$ is the set of possible colors (remember that the $Q_s$'s that are the $U_j$'s for $j\in [m]$ form a partition of the domain). The complete disjunctive normal form ~\refeq{eq:complete disjunction} of $\Psi_i(x)$ can be written into the suggestive form 
\[
\Ou_{(\mathbf{\eps}, \mathbf{\eps'})\in\Delta_i}\big(\,\Theta_{\mathbf{\eps}}(x) \et \Theta_{\mathbf{\eps'}}(\suc_i(x))\,\big).
\]
If each subformula $\textsf{m}_i(x)$ and $\textsf{M}_i(x)$ of $\phi$ is similarly put into complete disjunctive normal form, that is $\Ou_{(\sharp, \mathbf{\eps})\in\Delta_i} \Theta_{\mathbf{\eps}}(x)$ and $\Ou_{(\mathbf{\eps}, \sharp)\in\Delta_i} \Theta_{\mathbf{\eps}}(x)$, respectively (there is no ambiguity in our implicit definition of the $\Delta_i$'s, since $\sharp\not\in\Gamma$),  then the above sentence~\refeq{eq:localisation} equivalent to $\phi$ becomes the following equivalent sentence:
%$$
\[
\begin{array}{c}
\phi' = \ex\bold{U}\fa x{\dst \Et_{i\in [d]}} %\\
\left\{\begin{array}{rcll}
\min_i(x) 		& \imp & {\dst\Ou_{(\sharp,\mathbf{\eps})\in\Delta_i}         \Theta_{\mathbf{\eps}}(x)} & \et \\
\max_i(x) 		& \imp & {\dst\Ou_{(\mathbf{\eps},\sharp)\in\Delta_i}         \Theta_{\mathbf{\eps}}(x)} & \et \\
\neg \max_i(x) & \imp & {\dst\Ou_{(\mathbf{\eps},\mathbf{\eps'})\in\Delta_i} \left(\, \Theta_{\mathbf{\eps}}(x) \et\Theta_{\mathbf{\eps'}}(\suc_i(x)) \,\right)}
\end{array}\right\}
\end{array}
\]
%$$
Finally, let $L_{\local}$ denote the $d$-language over $\Gamma$ defined by the first-order sentence $\phi_{\local}$ obtained by replacing each $\Theta_{\mathbf{\eps}}$ by the new unary relation symbol $Q_{\mathbf{\eps}}$ in the first-order part of $\phi'$. In other words, $\picode{d}(L_{\local})$ is defined by the following first-order sentence:
\[
\begin{array}{c}
\phi_{\local} = \fa x\Et_{i\in [d]} %\\
\left\{\begin{array}{rcll}
\min_i(x) 		& \imp & {\dst\Ou_{(\sharp,\mathbf{\eps})\in\Delta_i}         Q_{\mathbf{\eps}}(x)} & \et \\
\max_i(x) 		& \imp & {\dst\Ou_{(\mathbf{\eps},\sharp)\in\Delta_i}         Q_{\mathbf{\eps}}(x)} & \et \\
\neg \max_i(x) & \imp & {\dst\Ou_{(\mathbf{\eps},\mathbf{\eps'})\in\Delta_i} \left(\, Q_{\mathbf{\eps}}(x) \et Q_{\mathbf{\eps'}}(\suc_i(x)) \,\right)}
\end{array}\right\}
\end{array}
\]
Hence, $L_{\local}=L(\Delta_1,\dots,\Delta_d)$. That is, $L_{\local}$ is indeed local and the corresponding sets of tiles are the $\Delta_i$'s of the previous formula.
It is now easy to see that our initial $d$-language $L$ is the projection of the local language $L_{\local}$ by the projection $\pi:\Gamma\imp\Sigma$ defined as follows: $\pi(\eps)=s$ iff $\eps_i=1$ for $i\in [m]$ and $U_i$ is $Q_s$. This completes the proof.
\end{proofsketch}

% proposition
\begin{proposition}\label{thm:eso(1)=eso(1,1) on pixel} 
$\eso{}(\arity {1})\subseteq\eso{}(\fa^{1},\arity 1)$ on pixel structures, for any $d>0$.
\end{proposition}

% proofsketch
\begin{proofsketch}
In a pixel structure, each function symbol $\mathrm{succ}_i$ is interpreted as a cyclic successor, that is a bijective function. So, a pixel structure is a \emph{bijective structure}, that is a first-order unary structure whose (unary) functions $f$ are bijective and that explicitly includes all their inverse bijections $f^{-1}$. It has been proved in~\cite{DurandG06} that any first-order formula on a bijective structure can be rewritten as a so-called \emph{cardinality formula}, that is as a boolean combination of sentences of the form $\psi^{\ge k}=\ex^{\ge k}x\;\psi(x)$ (for $k\ge 1$) where $\psi(x)$ is a quantifier-free formula with the single variable $x$ and where the quantifier $\ex^{\ge k}x$ means "there exist at least $k$ elements $x$". 
Therefore, it is easily seen that proving the proposition amounts to show that each sentence of the form $\psi^{\ge k}$ or $\neg\psi^{\ge k}$ can be translated in $\eso{}(\fa^{1},\arity 1)$ on pixel structures. 

This is done as follows: for a given sentence $\ex^{\ge k}x\; \psi(x)$, we introduce new unary relations $U^{=0},$ $U^{=1},\dots$,$U^{=k-1}$ and $U^{\geq k}$, with the intended meaning: 
\begin{quote}
\textit{A pixel $a\in\dom n^d$ belongs to $U^{=j}$ (resp. $U^{\geq k}$) iff there are exactly $j$ (resp. at least $k$) pixels $b\in\dom n^d$ lexicographically smaller than or equal to $a$ such that $\picode{d}(\apic)\models\psi(b)$.}
\end{quote} 
Then we have to compel these relation symbols to fit their expected interpretations, by means of a first-order formula with a single universally quantified variable. First, we demand the relations to be pairwise disjoint: 
\begin{enumerate}[$(1)$]
\item  ${\dst \Et_{i<j<k} }\left(\neg U^{=i}(x)\ou\neg U^{=j}(x)\right) \et {\dst \Et_{i<k}}\left(\neg U^{=i}(x)\ou\neg U^{\ge k}(x)\right).$ 
\end{enumerate}

\newcommand{\inflex}{\leq_{\text{lex}}}
\newcommand{\suclex}{\suc_{\text{lex}}}
\newcommand{\minlex}{\min_{\text{lex}}}
\newcommand{\maxlex}{\max_{\text{lex}}}
Then, we temporarily denote by $\inflex$ the lexicographic order on $\dom n^d$ inherited from the natural order on $\dom n$, and by $\suclex$, $\minlex$, $\maxlex$ its associated successor function and unary relations corresponding to extremal elements. Then the sets described above can be defined inductively by the  conjunction of the following six formulas:  
\begin{enumerate}[$(1)$]\setcounter{enumi}{1}
\item  $\left( \minlex(x)\et\neg\psi(x) \right)  \imp  U^{=0}(x)$ \\
\item  $\left( \minlex(x)\et\psi(x) \right)      \imp   U^{=1}(x)$ \\  
\item  ${\dst\Et_{i<k}} \left( \neg \maxlex(x) \et U^{=i}(x) \et \neg \psi(\suclex(x)) \right)  \imp   U^{=i}(\suclex(x))$ \\
\item  ${\dst\Et_{i<k-1}} \left( \neg \maxlex(x) \et U^{=i}(x) \et \psi(\suclex(x)) \right)\imp U^{=i+1}(\suclex(x))$ \\
\item  $\left( \neg \maxlex(x) \et U^{=k-1}(x) \et \psi(\suclex(x)) \right)  \imp   U^{\ge k}(\suclex(x))$ \\
\item  $\left( \neg \maxlex(x) \et U^{\ge k}(x) \right)  \imp   U^{\ge k}(\suclex(x))$.
\end{enumerate}
Hence, under the hypothesis $(1)\et\dots\et (7)$, the sentences $ \psi^{\ge k}$ and $\neg \psi^{\ge k}$ are equivalent, respectively, to $\fa x (\maxlex(x)\imp U^{\ge k}(x))$ and $\fa x (\maxlex(x)\imp \neg U^{\ge k}(x)).$ 

To complete the proof, it remains to get rid of the symbols $\suclex$, $\minlex$ and $\maxlex$ that are not allowed in our language. It is done by referring to these symbols implicitly rather than explicitly. For instance, since $\suclex(x)=\suc_i\suc_{i+1}\dots\suc_d(x)$, for each non maximal $x\in[n]^d$, i.e., distinct from $(n,\ldots,n)$, and for the smallest $i\in [d]$ such that $\Et_{{j>i}}\max_j(x)$, each formula $\phi$ involving $\suclex(x)$ actually corresponds to the conjunction: 
\[
\Et_{i\in [d]}\left( (\neg\max_i(x)\et\Et_{i<j\leq d}\max_j(x)) \imp\phi_{i} \right),
\] 
where $\phi_{i}$ is obtained from $\phi$ by the substitution $\suclex(x)\devient \suc_i\suc_{i+1}\dots\suc_d(x)$. Similar arguments allow to get rid of $\minlex$ and $\maxlex$. 
\end{proofsketch}

% remark
\begin{remark}
In this proof, two crucial features of a structure of type $\picode{d}(\apic)$ are involved: 
\begin{itemize}
\item  its "bijective'' nature, that allows to rewrite first-order formulas as cardinality formulas with a single first-order variable; 
\item  the "regularity" of its predefined arithmetics (the functions $\suc_{i}$ defined on each dimension), that endows $\picode{d}(\apic)$ with a grid structure: it enables us to implicitly define a linear order of the whole domain $dom(\apic)$ by means of first-order formulas with a single variable, which in turn allows to express cardinality formulas by ``cumulative''  arguments, \textit{via} the sets $U^{=i}$. 
\end{itemize}
Proposition~\ref{thm:eso(1)=eso(1,1) on pixel} straightforwardly generalizes to all structures -- and there are a lot -- that fulfill these two properties. 
\end{remark}

To conclude this section, let us mention that we could rather easily derive from Theorem~\ref{thm:rec ssi eso(fa 1) ssi eso(fa 1,ar 1)} the following additional characterization of $\rec{d}$:

% corollary
\begin{corollary}\label{thm:rec ssi eso(fa 1) ssi eso(var 1)}
For any $d>0$ and any $d$-language $L$, the following assertions are equivalent:
\begin{enumerate}
\item $L\in \rec{d}$;
\item $\picode{d}(L)\in\eso{}(\var 1)$. 
\end{enumerate}
\end{corollary}

% SECTION % SECTION % SECTION % SECTION % SECTION % SECTION % SECTION % SECTION % SECTION % SECTION % SECTION % SECTION % SECTION % SECTION % SECTION % SECTION % SECTION % SECTION % SECTION % SECTION % SECTION % SECTION 
% SECTION % SECTION % SECTION % SECTION % SECTION % SECTION % SECTION % SECTION % SECTION % SECTION % SECTION % SECTION % SECTION % SECTION % SECTION % SECTION % SECTION % SECTION % SECTION % SECTION % SECTION % SECTION 
% section : CARAC NLIN
\section{A logical characterization of $\nlinca$}\label{sec:2 carac nlin}

Besides the notion of recognizable picture language, the main concept studied in this paper is the linear time complexity on nondeterministic cellular automaton of any dimension. 

% definition
\begin{definition}\titledef{cellular automaton}
A \textdef{one-way $d$-dimensional cellular automaton} (\textdef{$d$-automaton}, for short) over an alphabet $\Sigma$ is a tuple $\A =(\Sigma,\Gamma,\delta,F)$, where 
\begin{itemize}
\item  the finite alphabet $\Gamma$ called the  \textdef{set of states} of $\A $ includes the \textdef{input alphabet} $\Sigma$ and the set $F$ of \textdef{accepting states}: $\Sigma,F\subseteq\Gamma$;
\item  $\delta$ is the (nondeterministic) \textdef{transition function} of $\A $: $\delta:\Gamma\times (\Gamma^{\sharp})^{d}\to\mathcal{P}(\Gamma).$
\end{itemize}
\end{definition}

% definition
\begin{definition}\titledef{Successors of a picture for a cellular automaton}
Let $\A =(\Sigma,\Gamma,\delta,F)$ be a $d$-automaton and $\apic,\apic': [n]^{d}\to \Gamma$ be two $d$-pictures on $\Gamma$. We say that $\apic'$ is a \textdef{successor} of $\apic$ for $\A $, denoted by \textdef{$\apic'\in \A (\apic)$}, if for each $a\in [n]^{d}$, $$\apic'(a) \in \delta(\apic(a),\apic^{\sharp}(a^{(1)}),\dots, \apic^{\sharp}(a^{(d)})),$$ where $\apic^{\sharp}$ denotes the bordered picture of $\apic$ and, for $i\in[d ]$, $a^{(i)}$ denotes the point having the same coordinates as point $a$ except the $i^{th}$ one which is $a_i+1$. So, the value of $\apic'$ at point $a\in [n]^{d}$ only depends on the set of points $a,a^{(1)},\dots,a^{(d)}$ of $\apic$. This set of $d+1$ points is called the \textdef{neighborhood} of $a$.

The set of \textdef{$j^{th}$-successors} of $\apic$ for $\A$, denoted \textdef{$\A ^j(\apic)$}, is defined inductively: 
\[
\A^0(\apic)= \{\apic\} \text{ and for $j\geq 0$: } \A^{j+1}(\apic)=\bigcup_{\apic'\in\A^j(\apic)}\A (\apic').
\]
\end{definition}

% definition
\begin{definition}\titledef{computation and configuration of a cellular automaton}
A \textdef{computation} of a $d$-automaton $\A =(\Sigma,\Gamma,\delta,F)$ on an input $d$-picture $\apic$ on $\Sigma$ is a sequence $\apic_1$, $\apic_2$, $\apic_3$, \dots of $d$-pictures on $\Gamma$ such that $\apic_1=\apic$ and $\apic_{i+1}\in\A (\apic_i)$ for each $i$. The picture $\apic_i$, $i\geq 1$, is called the \textdef{$i^{th}$ configuration} of the computation. A computation is \textdef{accepting} if it is finite -- it has the form $\apic_1,\apic_2,\dots,\apic_k$ for some $k$ -- and the cell of minimal coordinates, $1^d=(1,\dots,1)$, of its last configuration is in an accepting state: $\apic_{k}(1^d)\in F$. 
\end{definition}

% remark
\begin{remark}
Note that the space used by a $d$-automaton is exactly the space (set of cells) occupied by its input $d$-picture.%
\end{remark}

% definition
\begin{definition}\titledef{Acceptation by a cellular automaton, linear time}
Let $\A =(\Sigma,\Gamma,\delta,F)$ be a $d$-automaton and let $T: \mathbb{N}\to\mathbb{N}$ be a function such that $T(n)>n$. A $d$-picture $\apic:[n]^d\to\Sigma$ is \textdef{accepted by $\A$ in time $T$} if $\A$ admits an accepting computation of length $T(n)$ on $\apic$. That means, there exists a computation $\apic=\apic_1,\apic_2,\dots,\apic_{T(n)}$ of $\A $ on $\apic$ such that $\apic_{T(n)}(1^d)\in F$. 

A $d$-language $L$ on $\Sigma$ is \textdef{accepted}, or \textdef{recognized}, by $\A$ \textdef{in time} $T$ if it is the set of $d$-pictures accepted by $\A$ \emph{in time} $T$. That is 
\[
L=\{\apic: \ex \apic'\in \A^{T(n)-1}(\apic)\ST \apic'(1^d)\in F\ \text{ where } dom(\apic)=[n]^d\}.
\] 
If $T(n)= cn+c'$, for some integers $c,c'$, then $L$ is said to be \textdef{recognized in linear time} and we write \textdef{$L\in\nlinca^{d}$}. 
\end{definition}

The time bound $T(n)>n$ of the above definition is necessary and sufficient to allow the information of any pixel of $\apic$ to be communicated to the pixel of minimal coordinates, $1^d$.

% remark
\begin{remark}\label{rk:robustesse nlin}
The nondeterministic linear time class $\nlinca^{d}$ is very robust, \ie  is not modified by many changes in the definition of the automaton or in its time bound. In particular, the constants $c,c'$ defining the bound $T(n) = cn+c'$ can be fixed arbitrarily, provided $T(n)>n$. Typically, the class $\nlinca^{d}$ does not change if we take the \emph{minimal} time $T(n)=n+1$, called \emph{real time}.
\end{remark}

\medskip

As we have stated several logical characterizations of the class $\rec{d}$ of recognizable picture languages (Theorem~\ref{thm:rec ssi eso(fa 1) ssi eso(fa 1,ar 1)}), we now give two characterizations of the class $\nlinca^{d}$ in existential second-order logic. 

% theorem
\begin{theorem}\label{thm:nlin ssi eso(var d+1) ssi eso(fa d+1,ar d+1)}
For any $d>0$ and any $d$-language $L$, the following assertions are equivalent:
\begin{enumerate}
\item\label{item:nlin}  $L\in\nlinca^{d}$;
\item\label{item:fa d+1, ar d+1}  $\cocode{d}(L)\in\eso{}(\fa^{d+1},\arity d+1)$;
\item\label{item:var d+1}  $\cocode{d}(L)\in\eso{}(\var d+1)$.
\end{enumerate}
\end{theorem}

This theorem is a straightforward consequence of Proposition~\ref{thm:nlin to eso - phase1} and~\ref{thm:eso(1)=eso(1,1) on pixel} below. The former states the equivalence $(\ref{item:nlin})\Ssi(\ref{item:fa d+1, ar d+1})$; with the latter, we establish the normalization $\eso{}(\var {d})\subseteq\eso{}(\fa^{d},\arity d)$ on coordinate structures, which yields equivalence $(\ref{item:fa d+1, ar d+1})\Ssi (\ref{item:var d+1})$. 

% proposition
\begin{proposition}\label{thm:nlin to eso - phase1} 
For any $d>0$ and any $d$-language $L$, $$L\in\nlinca^{d}\Ssi\cocode{d}(L)\in\eso{}(\fa^{d+1},\arity d+1).$$ 
\end{proposition}

% proofsketch
\begin{proofsketch}
% condition necessaire
\framebox{$\Imp$} 
Let $L\in\nlinca^{d}$. By Remark~\ref{rk:robustesse nlin}, let us assume without loss of generality that $L$ is recognized by a $d$-automaton $\mathcal{A}=(\Sigma,\Gamma,\delta,F)$ in time $n+1$. The sentence in $\esorel{}(\fa^{d+1}, \arity  d+1)$ that we construct is of the form \,$\ex(R_s)_{s\in\Gamma}\:\fa\tu x\:\psi(\tu x)$, \,where $\psi(\tu x)$  is a quantifier-free formula such that:
\begin{itemize}
\item  $\psi$ uses a list of exactly $d+1$ first-order variables $\tu x=(x_1,\dots, x_d, x_{d+1})$. Intuitively, the $d$ first ones represent the coordinates of any point in $dom(\apic)=[n]^d$ and the variable $x_{d+1}$ represents any of the first $n$ instants $t\in[ n]$ of the computation (the last instant $n+1$ is not explicitly represented); 
\item $\psi$ uses, for each state $s\in \Gamma$, a relation symbol $R_s$ of arity $d+1$. Intuitively, $R_s(a_1,\dots, a_d, t)$ holds, for any $a=(a_1,\dots, a_d)\in[ n]^d$ and any $t\in[ n]$, iff the state of cell $a$ at instant $t$ is $s$. 
\end{itemize}
$\psi(\tu x)$ is the conjunction $\psi(\tu x)=\initfle(\tu x)\et\stepfle(\tu x)\et\laststepfle(\tu x)$ of three formulas whose intuitive meaning is the following.
\begin{itemize}
\item  $\forall \tu x \;\initfle(\tu x)$ describes the first configuration of $\mathcal{A}$, \ie  at initial instant $1$; 
\item  $\forall \tu x\; \stepfle(\tu x)$ describes the computation of $\mathcal{A}$ between the instants $t$ and $t+1$, 
for $t\in[ n-1]$, \ie  describes the $(t+1)^{th}$ configuration from the $t^{th}$ one;
\item  $\forall \tu x\; \laststepfle(\tu x)$ expresses that the $n^{th}$ configuration of $\mathcal{A}$ leads to a (last) $(n+1)^{th}$ configuration which is accepting, \ie  with an accepting state in cell $1^d$.
\end{itemize}

Let us give explicitly the main formula $\stepfle$. The main technical difficulty comes from the description of the transition function $\delta$ of $\mathcal{A}$ for any cell in the border of the picture. In order to describe uniformly $\delta$ we introduce the notion of \emph{neighborhood}. A \textdef{neighborhood} for $\mathcal{A}$ is any function $\nu:[ 0,d] \imp \Gamma^{\sharp}$ such that $\nu(0)\in\Gamma$. Let \textdef{$\voisins(\mathcal{A})$} denote the set of neighborhoods $\nu = (\nu(0), \nu(1), \dots, \nu(d))$ of $\mathcal{A}$. So, $\delta$ is a function of prototype $\voisins(\mathcal{A}) \imp \mathcal{P} (\Gamma).$ 

Clearly, the universal closure $\forall \tu x \stepfle(\tu x)$ of the following formula  $\stepfle(\tu x)$ --- that uses the "neighborhood" notation -- correctly describes the computation between the instants $t$ and $t+1$, 
for $t\in[ n-1]$:
\[
\begin{array}{c}
\stepfle(\tu x)= %\\  %
{\dst \Et_{\nu\in\voisins(\mathcal{A})}}
\left\{
\begin{array}{ll}
\neg \max(x_{d+1}) \et R_{\nu(0)}(\tu x) & \et \vphantom{{\dst \Et_{i}}} \\
{\dst \Et_{i\in [ d]: \nu(i)\in \Gamma}}  \left( \neg \max(x_{i}) \et R_{\nu(i)}(\tu x^{(i)}) \right) & \et \\
{\dst \Et_{i\in[ d]: \nu(i)=\sharp}}  \max(x_i) 
\end{array}
\right\}
%  \\ %
\imp {\dst \bigoplus_{s\in\delta(\nu)}} R_{s}(\tu x^{(d+1)}). 
\end{array}
\]
Here, $\bigoplus$ denotes the exclusive disjunction and $\tu x^{(i)}$ is the $(d+1)$-tuple $\tu x$ where $x_i$ is replaced by $succ(x_i)$. Hence we have proved that for any $d$-picture $\apic$ on $\Sigma$, $\cocode{d}(\apic)$ satisfies the $\esorel{}(\fa^{d+1},\arity d+1)$-sentence \,$\ex(R_s)_{s\in \Gamma} \fa \tu x\;\psi(\tu x)$\, if, and only if, $\mathcal{A}^{n}(\apic)$ contains an accepting configuration, \ie $\mathcal{A}$ accepts $\apic$ in time $n+1$, or, by definition, $\apic\in L$. Hence $L\in \esorel{}(\fa^{d+1}, \arity  d+1)$, as required.

% condition suffisante
\framebox{$\Rimp$} 
Assume $\cocode{d}(L)\in\eso{}(\fa^{d+1},\arity d+1)$, \ie there is some sentence $\Phi\in\eso{}(\fa^{d+1},\arity d+1)$ such that $\apic\in L \Ssi \cocode{d}(\apic)\models \Phi$. We want to prove $L\in\nlinca^{d}$, \ie $L$ is recognized by some $d$-automaton in linear time. Let us give the main idea of the proof for the simplest case $d=1$ and a formula $\Phi\in\eso{}(\fa^{2},\arity 2)$ of the form $$\Phi=\ex R\: \fa x \fa y\: \psi(x,y)$$ where $R$ is a binary relation symbol and $\psi$ is a quantifier-free formula where the only atoms in which $R$ occurs, called $R$ atoms, are of the following forms (1-4): 
\begin{center}
\hspace{1ex} \hfill (1) $R(x,y)$; \hfill (2) $R(succ(x),y)$; \hfill (3) $R(x, succ(y))$; \hfill (4) $R(y,x)$. \hfill \hspace{1ex}
\end{center}
First, notice that if the only atoms where $R$ occurs are of the forms (1-3), \ie the variables $x,y$ only appear in this \emph{unique} order in the arguments of $R$, then formula $\Phi$ has a \emph{local behaviour}: points  
$(x,y)$, $(succ(x),y)$ and $(x,succ(y))$ are \emph{neighbours}, \ie adjacent each other. This allows to construct a $1$-automaton (nondeterministic cellular automaton of dimension $1$) $\mathcal{A}$ that "mimics" $\Phi$.
Roughly, $\mathcal{A}$ "guesses" successively "rows" $R(i,\ldots)$, for $i=1,2,\ldots n$, of $R$, and in the same time, it "checks locally" the coherence of each "instantiation" $\psi(i,j)$: more precisely, at "instant" $i$,
the state of each cell $j$, $1\le j \le n$, of $\mathcal{A}$ "contains" both values $R(i,j)$ and $R(i+1,j)$. So, in case $R$-atoms are of the forms (1-3), the language $L$ is recognized by such a $1$-automaton  $\mathcal{A}$ in linear time as claimed.

Now, let us consider the "general" case where the formula includes all the forms (1-4). Of course, the "pixel" (4) $R(y,x)$ is not adjacent to "pixels" (1-3) but it is their symmetric (more precisely, it is the symmetric of $R(x,y)$) with respect to the "diagonal" $x=y$. The intuitive idea is "to cut" or "to fold" the "picture" $R$ along this diagonal: $R$ is replaced by its two "half pictures" denoted $R_1$ and $R_2$, that are "superposed" in the "half square" $x\le y$ above the diagonal. More precisely, $R_1$ and $R_2$ are binary relations whose intuitive meaning is the following: for points $(x,y)$ such that $x\le y$, one has the equivalence $R_1(x,y)\Ssi R(x,y)$ and the equivalence $R_2(x,y)\Ssi R(y,x)$. By this transformation, each pixel 
$R_2(x,y)$ that represents the original pixel $R(y,x)$ lies at the same point $(x,y)$ as pixel $R_1(x,y)$ that represents pixel $R(x,y)$, for $x\le y$ (the case $y\le x$ is similar). This solves the problem of vicinity.

More precisely, the sentence $\Phi=\ex R \:\fa x \fa y\: \psi(x,y)$ is normalized as follows. Let $\textsf{coherent}(x,y)$ denote the formula $$x=y\imp (R_1(x,y)\ssi R_2(x,y))$$ whose universal closure ensures the coherence of $R_1$ and $R_2$ on the common part of $R$ they both represent, that is the diagonal $x=y$. Using $R_1$ and $R_2$, it is not difficult to construct a formula 
$$
\psi'(x,y) =\left( 
\begin{array}{ll}
\textsf{coherent}(x,y) & \et \\
x<y \to \psi_{<}(x,y) & \et \\ 
x=y \to \psi_{=}(x,y) & \et \\ 
x>y \to \psi_{>}(x,y)
\end{array}\right)
$$
such that the sentence $\Phi'=\ex R_1 \ex R_2\: \fa x \fa y \:\psi'(x,y)$ in $\eso{}(\fa^{2},\arity 2)$ is equivalent to $\Phi$. Let us describe and justify its precise form and meaning.

\begin{table}[htdp]
\caption{Replacement of $R$-atoms by $R_1$- or $R_2$-atoms}\label{table:case replacement of R}
\begin{center}
\hspace{-1em}
\begin{tabular}{|c|c||c|c|c|c|}
\hline
case & formula 		& $R(x,y)$ 	& $R(\suc(x),y)$ &$R(x,\suc(y))$ & $R(y,x)$ \\
\hline
\hline
 $x<y$& $\psi_{<}(x,y)$ & $R_1(x,y)$ & $R_1(\suc(x),y)$ &$R_1(x,\suc(y))$ & $R_2(x,y)$ \\
\hline
$x=y$& $\psi_{=}(x,y)$ & $R_1(x,y)$ & $R_2(x,\suc(y))$ & $R_1(x,\suc(y))$& $R_1(x,y)$ \\
\hline
$x>y$& $\psi_{>}(x,y)$ & $R_2(y,x)$ & $R_2(y, \suc(x))$ &$R_2(\suc(y),x)$ & $R_1(y,x)$ \\
\hline
\end{tabular}
\end{center}
\label{default}
\end{table}%

The formulas $\psi_{<}(x,y)$, $\psi_{=}(x,y)$ and $\psi_{>}(x,y)$ are obtained from formula $\psi(x,y)$ by substitution of $R$ atoms by $R_1$ or $R_2$ atoms according to the cases described in Table~\ref{table:case replacement of R}. It is easy to check that each replacement is correct according to its case. For instance, it is justified to replace each atom of the form $R(x,\suc(y))$ in $\psi$  by $R_2(\suc(y), x)$ when $x>y$ (in order to obtain the formula $\psi_{>}(x,y)$) because when $x>y$, then $\suc(y)\le x$ and hence the equivalence 
$R(x,\suc(y)) \ssi R_2(\suc(y), x)$ holds, by definition of $R_2$. 

Notice that the variables $x,y$ always occur in this order in each $R_1$ or $R_2$ atom of the formulas $\psi_{<}$ and $\psi_{=}$ (see Table~\ref{table:case replacement of R}). At the opposite, they always occur in the reverse order $y,x$ in the formula $\psi_{>}(x,y)$. This is not a problem because, by symmetry, the roles of $x$ and $y$ can be exchanged and the universal closure $\forall x \forall y (x>y \to \psi_{>}(x,y))$ is trivially equivalent to $\forall x \forall y (y>x \to \psi_{>}(y,x))$. So, the above sentence $\Phi'$ -- and hence, the original sentence $\Phi$ -- is equivalent to the sentence denoted $\Phi"$ obtained by replacing in $\Phi'$ the subformula $x>y\imp\psi_{>}(x,y)$ by $y>x\imp\psi_{>}(y,x)$. By construction, relation symbols $R_{1}$, $R_{2}$ only occur in $\Phi$ in atoms of the three required "sorted" forms: $R_i(x,y)$, $R_i(\suc(x),y)$ or $R_i(x, \suc(y))$. Finally, to be precise, there remain two difficulties so that a $1$-automaton can simulate the "sorted" sentence $\Phi"$ in linear time, by the informal algorithm described above: 
\begin{itemize}
\item the presence of equalities and inequalities in the sentence;
\item the forms of the atoms involving input relation symbols.
\end{itemize}
It is easy to get rid of equalities and inequalities by introducing new binary relation symbols defined and used in a "sorted" manner too (see (1-3)). Concerning the second point, we can assume, without loss of generality, that the only atoms involving the input relation symbols $(Q_s)_{s\in\Sigma}$ are of the two forms $Q_s(x)$ or $Q_s(y)$. As we do for equalities and inequalities, we can get rid of atoms of the form $Q_s(y)$ by introducing new binary $\eso{}$ relation symbols : intuitively, they convey each bit $Q_s(a)$ at each point of coordinates $(a,\ldots)$ or $(\ldots, a)$; those new binary relations are also defined and used in a "sorted" manner. The fact that all the atoms involving the input are of the form $Q_s(x)$ allows to consider this input in the initial configuration of the computation of the $1$-automaton \emph{but in no later configuration} as required. So, the sketch of proof is complete for the case $d=1$. 

For the general case, \ie for any dimension $d$, the ideas and the steps of the proof are exactly the same as for $d=1$ but the notations and details of the proof are much more technical. To give an idea, let us succinctly describe the $\eso{}$ relations of arity $d+1$ introduced in the main normalization step. Here again, each $\eso{}$ relation symbol $R$ of the original sentence $\Phi$ in $\eso{}(\fa^{d+1},\arity d+1)$ is replaced by -- or, intuitively, "divided into" -- $(d+1)!$ new $\eso{}$ relation symbol $R_{\alpha}$ of the same arity $d+1$, where $\alpha$ is a permutation of the set of indices $[d+1]$. The intended meaning of each relation $R_{\alpha}$ is the following: for each tuple $(a_1,\ldots,a_{d+1})\in [n]^{d+1}$ such that $a_1\le a_2 \ldots \le a_{d+1}$, the equivalence $$R_{\alpha}(a_1,\ldots,a_{d+1}) \ssi R(a_{\alpha(1)},\ldots, a_{\alpha(d+1)})$$ holds. 
Then, we introduce a partition of the domain $[n]^{d+1}$ into subdomains, similar to the partition of the domain $[n]^2$ described above for $d=1$ into the diagonal $x=y$ and the two "half domains" over and under the diagonal $x<y$ and $x>y$, respectively. According to the case (i.e. subdomain of the partition), this allows to replace each $R$ atom in $\Phi$ by an atom of one of the two following "sorted" forms: 
\begin{center} $R_{\alpha}(\tu x)$;  $R_{\alpha}(\tu x^{(i)})$ \end{center}
where $\tu x=(x_1,\ldots,x_{d+1})$, $1\le i \le d+1$, and $\tu x^{(i)}$ is the tuple $\tu x$ where $x_i$ is replaced by $\suc(x_i)$. Finally, the equalities and inequalities are similarly eliminated in the sentence and we normalize it with respect to the input $d$-ary relations $(Q_s)_{s\in\Sigma}$ by using new $\eso{}$ relation symbols of arity $d+1$ to convey the input information: in the final "sorted" sentence all the $Q_s$ atoms are of the unique form $Q_s(x_1,\ldots,x_{d})$. For such a "sorted" $\eso{}(\fa^{d+1},\arity d+1)$-sentence $\Phi$, it is now easy to construct a $d$-automaton that generalizes the automaton described above in case $d=1$, and checks in linear time whether $\cocode{d}(\apic)\models\Phi$. 
\end{proofsketch}

% proposition
\begin{proposition}\label{thm:eso(1)=eso(1,1) on pixel} 
For any $d>0$, $\eso{}(\var {d})\subseteq\eso{}(\fa^{d},\arity d)$ on coordinate structures. %$\coordstruc{d}$.
\end{proposition}

% proofsketch
\begin{proofsketch}
We first prove a kind of Skolemization of $\eso{}(\var {d})$-formulas, thus providing a first normalization of these formulas, in which the first-order part is \emph{universal} and includes \emph{the same number of first-order variables} than the initial formula. To illustrate the procedure that performs this preliminary normalization, let us run it on a very simple first-order formula with \emph{two} variables: 
% equation
\[%\begin{equation}\label{eq:formule a skolemiser}
\phi\equiv\ex x \left(\,\fa yU(x,y)\ou\ex yD(x,y)\,\right). 
\]%\end{equation} 
We introduce three new relation symbols $R_1,R_2,R_3$ corresponding to the quantified subformulas of~$\phi$. 
\[
\begin{array}{llll}
\defrel_1(R_1)\equiv\fa x:   & R_1(x) & \ssi & \fa y\, U(x,y) \\
\defrel_2(R_2)\equiv\fa x:   & R_2(x) & \ssi & \ex y\, D(x,y) \\
\defrel_3(R_3)\equiv\fa x:   & R_3(x) & \ssi & R_{1}(x)\ou R_{2}(x)
\end{array}
\]
Hence, our initial formula can be rewritten: 
% equation
\begin{equation}\label{eq:formule a skolemiser (2)}
\ex R_{1},R_{2},R_{3}: \left( \Et_{1\leq i\leq 3}\defrel_i(R_i) \right) \ \et \ex x\ R_3(x).
\end{equation} 
It is easily seen that \refeq{eq:formule a skolemiser (2)} can be written as a conjunction of prenex formulas, each of which involves no more than two variables and has a quantifier prefix of the shape $\fa x\fa y$ or $\fa x\ex y$ (we include in this latter form the subformula $\ex x\ R_3(x)$). 
All in all, $\phi$ is equivalent to a formula of the form: 
% equation
\begin{equation}\label{eq:formule a skolemiser (3)}
\ex R_{1},R_{2},R_{3}: \fa x\fa y\psi(x,y,\tu R)\et\fa x\ex y\theta(x,y,\tu R),
\end{equation} 
where $\psi$ and $\theta$ are quantifier-free. In order to put this conjunction under prenex form without adding a new first-order variable, we have to "replace" the existential quantifier by a universal one. (Afterward $\phi$, as a conjunction of formulas of prefix $\fa x,y$, could be written under the requisite shape.) To proceed, we get use of the arithmetics embedded in coordinate structures. It allows to defining a binary relation $\w{}$ with intended meaning: $\w{}(x,y)$ iff there exists $z\leq y$ such that $\theta( x,z)$ holds. This interpretation is achieved thanks to the formula: 
% equation
\begin{equation}\label{eq:witness relation}
\fa x, y\left\{
\begin{array}{ll}
\min(y)  \imp  \big( \w{}( x,y)\ssi\theta( x,y) \big) \hspace{1em}\et \\
\neg \max(y)\to (\w{}(x,\suc(y))\ssi\big( \theta( x,\suc(y))\ou \w{}( x,y)) \big)
\end{array}\right\}
\end{equation}
Under assumption~\refeq{eq:witness relation}, the assertion $\fa x\ex y\theta( x,y)$ is equivalent to $\fa x\fa y:\max(y)\imp\w{}( x,y)$. This allows to rewrite~\refeq{eq:formule a skolemiser (3)}, and hence $\phi$, as the formula:
% equation
\begin{equation}\label{eq:formule skolemisee}
\ex R_{1},R_{2},R_{3},\w{}\left(
\begin{array}{l}
(\ref{eq:witness relation})\ \et \fa x\fa y\psi(x,y,\tu R)\ \et \\
\fa x\fa y\big(\max(y)\imp\w{}( x,y)\big)
\end{array}\right)
\end{equation} 
which belongs to $\esorel{}(\fa^2)$.

\medskip

Thus, the above considerations allow to show the normalization $\esorel{}(\var d)=\esorel{}(\fa^d)$ on coordinate structures. It remains to prove $\esorel{}(\fa^d)=\esorel{}(\fa^d,\arity d)$. It amounts to build, for each formula $\Phi$ of type $\ex\tu R\fa x_{1},\dots,x_{d}\phi$, where $\phi$ is quantifier-free and $\tu R$ is a tuple of relation symbols of any arity, a formula $\Phi'$ with the same shape, but in which all relation symbols are of arity $\leq d$, such that $\Phi$ and $\Phi'$ have the same models, as far as coordinate structures are concerned. The possibility to replace a $k$-ary ($k\geq d$) relation symbol $R$ of $\Phi$ by $d$-ary symbols rests in the limitation of the number of first-order variables in $\Phi$: each atomic formula involving $R$ has the form $R(t_1,\dots,t_k)$ where the $t_i$'s are terms built on $x_1,\dots,x_d$. Therefore, although $R$ is $k$-ary, {\em in each of its occurrences} it behaves as a $d$-ary symbol, dealing with the sole variables $x_1,\dots,x_d$. Hence, the key is to create a $d$-ary symbol for each occurrence of $R$ in $\Phi$ or, more precisely, for each $k$-tuple of terms $(t_1,\dots,t_k)$ involved in a $R$-atomic formula. Let us again opt for a ``proof-by-example''  choice and illustrate the procedure on a very simple case. 

\medskip

Let $\Phi$ be the $\esorel{}(\fa^{2},\arity 3)$-formula $\ex R\fa x,y\phi(x,y,R)$, where $\phi\equiv R(x,y,x)\et\neg R(y,x,y)$. Introduce two new \emph{binary} relation symbols $R_{(x,y,x)}$ and $R_{(y,x,y)}$ associated to the triple of terms $(x,y,x)$ and $(y,x,y)$ involved in $\Phi$, and fix their interpretation as follows: for any $a,b\in\dom n$, \[R_{(x,y,x)}(a,b)\Ssi R(a,b,a)\AND R_{(y,x,y)}(a,b)\Ssi R(b,a,b).\] Then we get the equivalence:  \[\la S,R\ra\models\fa x,y: R(x,y,x)\et\neg R(y,x,y)\Ssi\la S,\tu R\ra\models\fa x,y: R_{(x,y,x)}(x,y)\et\neg R_{(y,x,y)}(x,y)\] which, in turn, yields the implication:
% equation
\begin{equation}\label{eq:phi -> tilde(phi)}
\begin{array}{c}
S\models\ex R\fa x,y : R(x,y,x)\et\neg R(y,x,y) \Imp %\\ 
S\models\ex\tu R\fa x,y : R_{(x,y,x)}(x,y)\et\neg R_{(y,x,y)}(x,y)
\end{array}
\end{equation} 
The converse implication would immediately complete the proof. Unfortunately, it does not hold, since the second formula has a model, while the first has not. 

To get the converse (right-to-left) implication of~\refeq{eq:phi -> tilde(phi)}, we have to strengthen the second formula with some assertion that compels the tuple $R_{(x,y,x)},R_{(y,x,y)}$ to be, in some sense, the binary representation of some ternary relation. This last construction is more sophisticated than the preceding ones, and we can't detail it here. 
\end{proofsketch}

%\input hierarchy.tex 

% section : hierarchy
\section{Hierarchy results}\label{sec:hierarchy}

In Section~\ref{sec:eso1 to rec}, we have established several natural characterizations of the classes $\rec{d}$ in existential second-order logic on pixel structures (Theorem~\ref{thm:rec ssi eso(fa 1) ssi eso(fa 1,ar 1)} and Corollary~\ref{thm:rec ssi eso(fa 1) ssi eso(var 1)}). Is there  similar characterizations in coordinate structures instead of pixel structures? One can trivially translate the logical characterization of  $\rec{d}$ in $\eso{}(\fa^1,\arity 1)$ for the  pixel representation into a characterization in $\esorel{}(\fa^{d},\arity d$) for the coordinate representation. 
However, one notices that the sentences so obtained in $\esorel{}(\fa^{d},\arity d$) are "sorted" in a sense similar to that evoked in the previous section (see the proof of Proposition~\ref{thm:nlin to eso - phase1}): the atoms involving the $\eso{}$ (resp. input) relation symbols, all of arity $d$, are of the "sorted" form $R(\tu x)$ or $R(\tu x^{(i)})$ (resp. $Q_s(\tu x)$) where $\tu x$ is the "sorted" tuple of variables $(x_1,\dots,x_d)$ and $\tu x^{(i)}$, $i\in[d]$, is the same "sorted" tuple where $x_i$ is replaced by $\suc (x_i)$. In fact, one can prove that a $d$-language $L$ is recognizable iff $\cocode{d}(L)$ is defined by a "sorted" sentence in $\esorel{}(\fa^{d},\arity d)$.
%Of course, by Lemma~\ref{thm:coord(Lplus) vs pixel(Lplus)} and Theorem~\ref{thm:rec ssi eso(fa 1) ssi eso(fa 1,ar 1)}, we get the equivalence, for every $d\ge 1$: 
% equation
%\begin{equation}\label{eq:rec ssi eso(d,d,s)}
%L\in\rec{d}\Ssi\cocode{d}(L)\in \esorel{}(\fa^{d},\arity d,\sorted).
%\end{equation}
But requiring the sentence to be sorted seems rather artificial. Can it be dropped? The following result shows it cannot.

% proposition
\begin{proposition}\label{thm:hierarchy}
For each integer $d\ge 2$ and for $d$-languages \emph{represented by their coordinate structures}, the following strict inclusions hold: 
$$\eso{}(\fa^{d-1},\arity d-1)=\esorel{}(\var d-1)\subsetneq\rec{d}\subsetneq\esorel{}(\var d)=\eso{}(\fa^{d},\arity d).$$ 
That means the following implications hold, 
\[
\begin{array}{c}
\cocode{d}(L)\in\esorel{}(\var d-1)\Imp L\in\rec{d}\Imp\cocode{d}(L)\in\esorel{}(\var d), 
\end{array}
\]
but neither of their converses does.
\end{proposition}

The proof of Proposition~\ref{thm:hierarchy} that we skip here uses a simple lemma about the so-called $\mirror{d}$ language.

%% definition
%\begin{definition}\titledef{\mirror{} language}\label{def:mirror} 
%For $d\ge 2$, we denote by \textdef{$\mirror{d}$} the $d$-language on $\Sigma=\{0,1\}$ defined as follows: a $d$-picture $\apic:[n]^d \to \{0,1\}$ belongs to $\mirror{d}$ iff for all $a=(a_1,\dots,a_d)\in [n]^d$, the equality $\apic(a)=\apic(a_{12})$ holds, where $a_{12}$ is the tuple $a$ in which components $a_1$ and $a_2$ have been exchanged. 
%\end{definition}

% definition
\begin{definition}\titledef{\mirror{} language}\label{def:mirror} 
For $d\ge 2$, we denote by \textdef{$\mirror{d}$} the set of pictures $\apic:[n]^d \to \{0,1\}$ that fulfill $\apic(a)=\apic(a_{12})$ for all $a\in [n]^d$. Here, $a_{12}$ is the tuple $a$ in which the first and second components have been exchanged (\eg, $(5,3,4,6)_{12}=(3,5,4,6)$). 
\end{definition}

% lemma
\begin{lemma}\label{thm:mirror}
For each $d\ge 2$, we trivially have $\cocode{d}(\mirror{d})\in \esorel{}(\var d)$. However, $\mirror{d} \not \in \rec{d}$. 
\end{lemma}
The proof of Lemma~\ref{thm:mirror} essentially amounts to a counting argument. It is a simple variant of the proof of a similar result in Giammarresi et al.~\cite{GRST96} and we skip it.

\newcommand{\Rep}{\text{Rep}}
\medskip

Finally, the hierarchy theorem below (Theorem~\ref{thm:hierarchy}) involves the number of universal first-order quantifiers. It uses the following "symmetric" language as a counterexample. 
\begin{definition}\titledef{symmetric language}\label{def:sym(d)}
Let $\sym{d}$ be the set of pictures $\apic :[n]^d \to \{0,1\}$ that fulfill, for every $a=(a_1,\dots,a_d)\in [n]^d$:
\begin{enumerate}
\item\label{item:sym permute}  $\apic (a)=\apic(a_{\alpha})$ for each permutation $\alpha$ of $\dom d$, where $a_{\alpha}=(a_{\alpha(1)},\dots,a_{\alpha(d)})$; 
\item\label{item:sym inverse}  $\apic (a)=\apic(a_{sym(i)})$ for all $i\in[d]$, where $a_{sym(i)}$ denotes the tuple $a$ whose $i^{th}$ component $a_i$ is replaced by its "symmetric value" $n+1-a_i$%
\footnote{Because of Item~\ref{item:sym permute}, Item~\ref{item:sym inverse} is clearly equivalent to the same assertion for only $i=1$.}. 
\end{enumerate}
In other words, the values $\apic (a)$ are defined up to all possible permutations of coordinates and up to all possible symmetries $a \mapsto a_{sym(i)}$.
\end{definition}

% lemma
\begin{lemma}\label{lemmaSYM}
For all $d\ge 2$, we have $\cocode{d}(\sym{d})\in \esorel{}(\fa^{d+1},\arity d) \setminus\esorel{}(\fa^{d},\arity d) .$
\end{lemma}

Here is our last hierarchy theorem.

% theorem
\begin{theorem}\label{thm:hierarchy}
For each integer $d\ge 2$ and for $d$-languages represented by \emph{coordinate} structures, the following (strict) inclusions hold:

\newcommand{\dd}{\hspace{-.8ex}&\hspace{-.8ex}}
\[
\begin{array}{ccccc}%{llllc}
\rec{d} \dd \stackrel{}{\subsetneq} \dd \esorel{}(\var d) \dd \stackrel{}{=} \dd \esorel{}(\fa^d,\arity d) \\
&&&& \downsubsetneq \\
&&&& \esorel{}(\fa^{d+1},\arity d) \\ 
&&&& \downsubseteq \\
\nlinca^{d} \dd \stackrel{}{=} \dd \esorel{}(\var d+1) \dd \stackrel{}{=} \dd \esorel{}(\fa^{d+1},\arity d+1) 
\end{array}
\]
\end{theorem}

The above strict inclusion $\esorel{}(\var d)\subsetneq\esorel{}(\fa^{d+1},\arity d)$ trivially yields the following result. 

% corollary
\begin{corollary}\label{corolstrictinclusion}
For each integer $d\ge 2$ and for $d$-languages represented by \emph{coordinate} structures, we have the \emph{strict inclusion}: $\esorel{}(\var d) \subsetneq \esorel{}(\arity d).$
\end{corollary}

% remark
\begin{remark}
Corollary~\ref{corolstrictinclusion} surprisingly contrasts with the equality $\esorel{}(\var 1)=\esorel{}(\arity 1)$ for $d$-languages represented by \emph{pixels} structures, see  Theorem~\ref{thm:rec ssi eso(fa 1) ssi eso(fa 1,ar 1)} and Corollary~\ref{thm:rec ssi eso(fa 1) ssi eso(var 1)}.
\end{remark}

\subsection*{Other hierarchies}
It is also natural to address the question whether the following "hierarchies" about second-order logic (SO) over picture languages are strict or not:
\begin{enumerate}
\item Is there a strict hierarchy of $\so$ or $\esorel{}$ according to the \emph{second-order quantifier alternation}?
\item Is there a strict hierarchy of the classes $\esorel{}(\arity d)$ and $\esorel{}(\fa^k,\arity d)$ according to the \emph{number of $\esorel{}$ relation symbols}?
\end{enumerate}
It is well-known that the answer to question 1 is negative for $\mso$ on word languages and tree languages: on these classes of languages, $\mso=\mathrm{EMSO}$ holds. That is, the hierarchy collapses at it first level.

At the opposite, Matz and Thomas~\cite{MatzT97} have answered question 1 positively for $\mso$ over $2$-dimensional picture languages: the quantifier alternation is strict for MSO on $2$-picture languages (in pixel representation). We strongly conjecture that their result still holds with a similar proof for any dimension $d>2$ and the pixel representation. We also conjecture that similar results hold in the coordinate representation: typically, we think that for $d$-languages with $d\ge 2$, the quantifier alternation of $\so(\arity d)$ gives a strict hierarchy. Unfortunately, we have no idea to prove it, even in the "minimal case", that is $\esorel{}(\arity 2)\subsetneq \mathrm{SO}(\arity 2)$, for $2$-languages.

In contrast, the answer to question 2 is totally known and uniform for picture languages of any dimension. In all cases, the hierarchy collapses at its first level. More precisely, Thomas~\cite{Thomas82} has established that every EMSO sentence over words is equivalent to a sentence whose monadic quantifier prefix consists of a single existential quantifier. Matz~\cite{Matz98} has proved the same result over $2$-pictures in the pixel representation. The proof of Matz can be extended (with slight adaptations) to any dimension $d$, for both pixel and coordinate representations. In other words, in both representations, all the logical classes -- essentially $\esorel{}(\arity d)$ and $\esorel{}(\fa^k,\arity d)$ --  
we have studied over picture languages of any dimension are not modified by the requirement there should be only \emph{one} ESO relation symbol.

\section*{Some remarks about the dimensions of pictures}

In this paper, for sake of simplicity and uniformity, we have chosen to restrict the presentation of our results to "square" pictures, \ie pictures of prototype $p:[n]^d\to\Sigma$. This may appear as a too strong requirement. In this section, we explain how our results about logical classes and their relationships with complexity classes $\rec{d}$ and $\nlinca^{d}$  
can be extended to the "most general" picture languages, i.e. to sets of $d$-pictures of prototype 
 $p:[n_1]\times\ldots\times[n_d]\to\Sigma$. "Most general" means as much general as they make sense in the logical or complexity theoretical framework involved.
 
 \subsection*{$\rec{d}$ and logical characterizations in pixel encodings} 
 All our results about logical characterizations of picture languages in pixel encodings and the class $\rec{d}$, \ie  results of Section~\ref{sec:eso1 to rec}, hold for pictures of general prototype $p:[n_1]\times\ldots\times[n_d]\to\Sigma$, without any restriction, \ie for all $n_i\ge 1$, $i\in[d]$. Moreover, our proofs also hold without change, except of course, the references to integer $n$ which should be replaced by integer $n_i$ according to the involved dimension $i\in[d]$.
 
 \subsection*{$\nlinca^{d}$ and logical characterizations in coordinate encodings} 
 The definition of linear time complexity of $d$-automata is not clear for input pictures whose domain has the general form $[n_1]\times\ldots\times[n_d]$ without restriction. However, linear time, \ie time $O(n)$, makes sense when the $n_i$ are of the same order $\Theta(n)$, \ie the pictures are "well-balanced". This justifies the following definition.
 
 \begin{definition} A $d$-picture language $L$ is \emph{well-balanced} if there is some constant positive integer $c$ such that the domain $[n_1]\times\ldots\times[n_d]$ of each picture $p\in L$ fulfills the condition: 
 for $i\in[d]$, $n\le c\,n_i$, where $n=\max(n_1,\ldots,n_d)$ is called the \emph{length} of $p$; we say that the picture $p$ (resp. the picture language $L$) is $c$-\emph{balanced}.
  \end{definition}
  
It is straightforward to adapt our notion of linear time to well-balanced picture languages.

\begin{definition} A $c$-balanced $d$-picture language $L$ belongs to $\nlinca^{d}$ if there exist a $d$-automaton $\A$ and a linear function $T(n)=c_1\,n+c_2$ such that $L$ is the set of $c$-balanced $d$-pictures $p$ accepted by $\A$ in time $T(n)$, where $n$ is the length of $p$.
\end{definition}
  
\begin{remark} 
Those notions are justified by the following points:
\begin{itemize}
\item The "perimeter" of a $c$-balanced $d$-picture $p$ of length $n$ is $\Theta(n)$ and its size (area, volume, etc., according to its dimension $d$) is $|p|=\Theta(n^d)$.
\item So, linear time means time \emph{linear in the perimeter} of the $d$-picture $p$ or, equivalently, in $|p|^{1/d}$.
\end{itemize}
\end {remark}

We can easily extend all our results about square $d$-languages to well-balanced $d$-languages. In order to reduce the well-balanced case to the square case we need some new definitions.
\begin{definition}
To each $d$-picture $p$ of length $n$, one associates its \emph{squared} $d$-picture denoted $p^=$ of domain $[n]^d$ obtained by putting the new special symbol $\square$ in each additional cell. Formally,
$$
\apic^=(a)=\left\{\begin{array}{l} \apic(a)\IF a\in dom(\apic) ; \\ \square \OTHERWISE.  \end{array}\right.
$$ 
To any $d$-picture language $L$ one associates its \emph{squared} $d$-picture language 
$$L^= =\{p^=: p\in L\}.$$
\end{definition}

\begin{remark} 
Let c be a constant. If $p$ is a $c$-balanced $d$-picture, then the size of its squared picture $p^=$ is 
$|p^=| =\Theta(|p|)$ 
\end {remark}

The following result is quite easy to prove.
\begin{lemma}~\label{lemma squared L}
Let $L$ be a $c$-balanced $d$-picture language. Then $L\in\nlinca^{d}\Ssi L^=\in\nlinca^{d}.$
\end{lemma}

\begin{remark} 
Notice that the nontrivial (right-to-left) implication of the Lemma~\ref{lemma squared L} means that the $d$-automaton $\A_1$ that recognizes $L^=$ in linear time can be simulated by some $d$-automaton $\A_2$ that recognizes $L$ in the same time (up to a constant factor) but with less space: $\A_2$ only uses the cells of $p$ instead of the cells of the squared picture $p^=$. This is possible with the following trick: the fact that the $i^{th}$ dimension $n$ of $p^=$ is replaced by $n_i\ge c^{-1}n$ allows to "fold" $c$ times the picture $p^=$ along the $i^{th}$ dimension $n_i$ of $p$. All in all, each cell of $p$ simulates (at most) $c^d$ cells of $p^=$. This is performed by taking for the set of states of $\A_2$ the set $\Gamma^{c^d}$ where $\Gamma$ is the set of states of $\A_1$.
\end{remark}

Now, let us compare any well-balanced $d$-picture language $L$ and its squared language $L^=$ from a logical point of view. The domain of the coordinate representation of a $d$-picture $p:[n_1]\times\ldots\times[n_d]$ is naturally $[n]$ where $n$ is the length of $p$, \ie $n=\max(n_1,\ldots,n_d)$. So, we define the coordinate representation of $p$ as $$\cocode{d}(p)=\cocode{d}(p^=).$$ 
Hence, as a trivial consequence of the definitions,
$$\cocode{d}(L)=\cocode{d}(L^=).$$

The previous definitions and remarks and Lemma~\ref{lemma squared L} justify that the results of Sections~\ref{sec:2 carac nlin} and~\ref{sec:hierarchy} hold without change in the extended case of well-balanced $d$-picture languages.

%\input conclusion.tex

% section : CONCLUSION
\section*{Conclusion}
\noindent
Notice that Theorem~\ref{thm:nlin ssi eso(var d+1) ssi eso(fa d+1,ar d+1)} that characterizes the \emph{linear} time complexity class of nondeterministic \emph{cellular automata} is very similar to the following result about time complexity $O(n^d)$, for any $d\ge1$, of nondeterministic \emph{RAMs}, by two of the present authors~\cite{GrOl04}: 
$$\ntimeram(n^d)=\esof{}(\var d)=\esof{}(\fa^d,\arity d).$$ 
The main difference is that this latter result involves the existential second-order logic \emph{with functions} ($\esof{}$) \emph{instead of} or \emph{in addition} to \emph{relations} and holds in all kinds of structures without restriction: pictures, structures of any arity and any type, \etc. It is also interesting and maybe surprising to notice that, in those results, the \emph{time degree} $d$ of a RAM computation plays the same role as the \emph{dimension} $d+1$ of the time-space diagram of a linear time bounded computation for a $d$-dimensional cellular automaton.

Both results attest of the robustness of the time complexity classes $\ntimeram(n^d)$ and $\nlinca^{d}$. They stress the significance of the RAMs as a sequential model, and of the cellular automata as a parallel model. 

\medskip

Finally, if we continue the comparison, we have, for coordinate representation of $d$-languages and for $d\geq 2$,  the (strict) inclusions: 
$$\rec{d}\subsetneq\eso{}(\fa^d,\arity d)\subsetneq\eso{}(\fa^{d+1},\arity d)\subseteq\nlinca^{d}$$ 
for cellular automata and $\eso{}$, which strongly contrast with the equalities: 
$$\ntimeram(n^d)=\esof{}(\fa^1,\arity 1)=\esof{}(\fa^2,\arity 1)$$
for RAMs and $\esof{}$~\cite{GrOl04,DurandGO04}.

% bibliographie
\bibliographystyle{alpha}
\bibliography{./0}

\newpage

\renewcommand{\margenote}[1]{}
\appendix
%\section
\section{Normalization of "guessed" relations}\label{sec:sort R}
%	\input normalHalfSorted.tex
%% sortR.tex

\providecommand{\infe}{\prec}%{Q_<}
\providecommand{\supe}{\succ}%{Q_>}
\providecommand{\egal}{\thickapprox}%{\egal}
\providecommand{\flip}[1]{\text{flip}(#1)}

This appendix is dedicated to the proof of the normalization of $\eso{}(\fa^d,\arity d)$ on coordinate encodings of $(d-1)$-pictures, for $d\geq 2$. 
Our purpose is to rewrite $\eso{}(\fa^d,\arity d)$-formulas into equivalent formulas that fulfill the "sorted" or "local" property mentioned in the proof of the "sufficient condition" of Proposition~\ref{thm:nlin to eso - phase1}.

Before formalizing our notion of "sorted" formulas, let us detail what is its intended meaning. We want to deal with pixels of time-space diagram of the computation that are adjacents, that is, that are both connected and differ (by one) by at most one dimension. Two such pixels are represented by $d$-tuples of the form $\tu x$ and $\tu x^{(i)}$, that is: 
\begin{itemize}
\item  their components are in the same order (elsewhere they could be disconnected); 
\item  there is at most one occurrence of $\suc$ (elsewhere, they would differ of more than one dimension). 
\end{itemize}
These requirements (and a little more) are formalized in the following definition.

% definition
\begin{definition}\label{def:sorted}
Let $k,d$ be two integers such that $d\ge k\ge 1$. A sentence over coordinate structures for $k$-pictures is in \textdef{$\esorel{}(\fa^{d},\arity  d,\sorted)$} if it is of the form 
$\ex \mathbf{R}  \fa \mathbf{x}  \psi(\mathbf{x})$
where
\begin{itemize}
\item $\mathbf{R}$ is a list of relation symbols of arity $d$;
\item $\psi$ is a quantifier-free formula whose list of first-order variables is $\mathbf{x}=(x_1,\dots,x_d)$;
\item each atom of $\psi$ is of one of the following forms:
	\begin{enumerate}[$(i)$]
	\item $Q_s(x_1,\dots, x_{k})$, for $s\in\Sigma$,
	\item $R(\mathbf{x})$ or $R(\mathbf{x}^{(i)})$ where $R\in \mathbf{R}$, $i\in[ d]$, 
	and $\mathbf{x}^{(i)}$ is the tuple $\mathbf{x}$ where $x_i$ is replaced by $succ(x_i)$,
	\item $min(x_i)$ or $max(x_i)$, for $i\in[d]$.
	\end{enumerate}
\end{itemize}
\end{definition}
We prove the normalization $\eso{}(\fa^d,\arity d)=\esorel{}(\fa^{d},\arity  d,\sorted)$ for $(d-1)$-pictures (\ie, for $k=d-1$). 
In the present section, we deal with Condition~$(ii)$ of the above definition (see Proposition~\ref{thm:normal-sort-R}). 
In Subsection~\ref{sec:arithmetic}, we eliminate equalities and inequalities. 
At this point, we get a normalization of $\eso{}(\fa^d,\arity d)$ into the so-called "half-sorted logic", denoted by $\esorel{}(\fa^{d},\arity  d,\semisorted)$. It remains to manage the input relation symbols; this is done in Section~\ref{sec:sort Q}, where we tackle Condition~$(i)$.

% subsection\subsection{Normalization of "guessed" relations}\label{sec:sort R}

To lighten the presentation of the forthcomming results, we first introduce some notations about tuples and permutations. 

% definition
\begin{definition}\label{def:tuples et permutations}\titledef{Permutation and tuples}
Let $n,d>0$ and $\tu x\in\dom n^d$. 
\begin{itemize}
\item  We denote by $[\tu x]_i$ the $i^{\text{th}}$ component of $\tu x$. \Eg $(5,7,2)_2=7$.
\item  We say that $\tu x$ is \textdef{non decreasing}, and we write \textdef{$\tu x\up$}, when $[\tu x]_1\leq\dots\leq [\tu x]_d$. 
\item  \textdef{$\perm(d)$} stands for the set of permutations of $\{1,\dots,d\}$. Given $\al_1,\dots,\al_d\in\{1,\dots,d\}$, we denote by $\al_1\dots\al_d$ the permutation $\al\in\perm(d)$ that maps each $i$ on $\al_i$. Conversely, for $\al\in\perm(d)$ we set $\al_i\egaldef \al(i)$. By \textdef{$\Transpo(d)$} we denote  the set of transpositions over $\{1,\dots,d\}$.
\item  If $\al\in\perm(d)$ and $\tu x=(x_1,\dots,x_d)$, we denote by \textdef{$\tu x_{\al}$} the $d$-tuple $(x_{\al_1},\dots,x_{\al_d})$. It is less ambiguous to define $\tu x_{\al}$ by the assertion: 
\begin{center}for any $i\in\{1,\dots,d\}$, $[\tu x_{\al}]_i=[\tu x]_{\al(i)}$.\end{center} 
Thus, if $\beta$ also belongs to $\perm(d)$, we get $[(\tu x_{\al})_{\beta}]_i=[\tu x_{\al}]_{\beta(i)}=[\tu x]_{\al\beta(i)}.$ 
Whence the identity: $(\tu x_{\al})_{\beta}=\tu x_{\al\beta}$. In particular, $(\tu x_{\al})_{\al^{-1}}=\tu x$.
\item  For $\al\in\perm(d)$ and $n>0$, we set $\textdef{$\sby{\al}$}=\{\tu x\in\dom n^d \st \tu x_{\al}\up\}$. In particular, denoting by $\id$ the identity on $\{1,\dots,d\}$, we get $\tu x\in\sby{\id}$ iff $\tu x\up$. Therefore, $\tu x\in\sby{\al}$ iff $\tu x_{\al}\in\sby{\id}$. Clearly, $\dom n^d=\bigcup_{\al\in\perm(d)}\sby{\al}$. 
\item  For any $i$ in $\{1,\dots,d\}$, $\tu x^{(i)}$ denotes the tuple obtained from $\tu x$ by replacing its $i^{\text{th}}$ component by its own successor. That is, if $\tu x=(x_1,\dots,x_d)$ then 
\[\tu x^{(i)}=(x_1,\dots,x_{i-1},\suc(x_{i}),x_{i+1},\dots,x_d).\] 
As previously, the arrangement of $\tu x^{(i)}$ according to some permutation ${\al}$ is denoted by $(\tu x^{(i)})_{\al}$. 
\end{itemize}
\end{definition}

% example
\begin{example} 
Consider $\tu x=(5,3,7,2)$ in $\dom 9^4$ and $\al=4213$, $\beta = 1432$ in $\perm(4)$. Then $\tu x_{\al}=(2,3,5,7)$ is non-decreasing while $\tu x_{\beta}=(5,2,7,3)$ is not. Besides, $\tu x^{(3)}=(5,3,8,2)$ while $(\tu x_{\al})^{(3)}=(2,3,6,7)$ and $(\tu x^{(3)})_{\al}=(2,3,5,8)=(\tu x_{\al})^{(4)}$.
\end{example}

\medskip

With these notations, the request described at the beginning of the section can be rephrased as follows: we want to normalize $\esorel{}(\fa^{d},\arity d)$-formulas in such a way that each atomic subformula $R(t_1(\tu x),\dots,t_p(\tu x))$ built with a guessed relation symbol $R$ has either the form $R(\tu x)$ or the form $R(\tu x^{(i)})$.

%% fact
%\begin{fact}\label{thm:x in [a] ssi x_b in [b-1a]}\margenote{\ref{thm:x in [a] ssi x_b in [b-1a]} est-il utile?}
%For $n,d>0$, $\tu x\in\dom n^d$ and $\al,\beta\in\perm(d)$: \ $\tu x\in\sby{\al}\Ssi\tu x_{\beta}\in\sby{\beta^{-1}\al}.$
%\end{fact}

% fact
\begin{fact}\label{thm:sans repetition}
On $\coordstruc{d-1}$,%
\footnote{We denote by $\coordstruc{d-1}$ the sets of coordinate encodings of $(d-1)$-pictures.} 
any formula $\Phi=\ex\tu R\fa\tu x\phi(\tu x,\tu R,\sg)\in\esorel{}(\fa^{d},\arity d)$ of signature $\sg$ can be written in such a way that: 
\begin{enumerate}[$(a)$]
\item\label{item:var_i cap var_j vide}  In each atomic subformula $R(t_1,\dots,t_p)$ of $\phi$, $\Var(t_i)\cap\Var(t_j)=\emptyset$ for every  $1\leq i<j\leq p$.
\item\label{item:arite exactement d}  Each R in $\tu R$ has arity $d$ exactly. 
\end{enumerate}
\end{fact}

% proof
\begin{proof}
The proof of \refitem{item:var_i cap var_j vide} is quite immediate. We illustrate it with an example: assume that $\Phi$ involves the subformula $R(\suc^2 x_1,x_2,\suc\,x_1,\suc^3x_2)$ for some $R\in\sg\cup\{\tu R\}$. Then clearly, $\Phi$ is equivalent, on picture-structures, to: \margenote{\`A v\'erifier}
\[\begin{array}[t]{c}
\ex\tu R\ex A\fa\tu x : \tilde{\phi}(\tu x,\tu R,A,\sg)\ \et %\\ 
(x_1=\suc\,x_3\et x_4=\suc^3x_2)\imp(A(x_2,x_3)\ssi R(x_1,x_2,x_3,x_4))
\end{array}\]
where $\tilde{\phi}$ is obtained from $\phi$ by substituting the formula $A(x_2,\suc(x_1))$ to each occurrence of the atom $R(\suc^2 x_1,x_2,\suc\,x_1,\suc^3x_2)$. 

\medskip

In order to prove~\refitem{item:arite exactement d}, assume for simplicity that $\tu R$ reduces to the single relational symbol $R$ of arity $p<d$. The idea is to add $d-p$ dummy arguments to $R$. Clearly, $\Phi$ is equivalent to the formula: 
\[\begin{array}{c}
\ex R:\ \tilde{\phi}\et %\\ 
\fa\tu x\Et_{p<i\leq d}R(x_1,\dots,x_i,\dots,x_d)\ssi R(x_1,\dots,\suc(x_i),\dots,x_d),
\end{array}\]
where $\tilde{\phi}$ is obtained from $\phi$ by replacing each atomic subformula $R(t_1,\dots,t_p)$ by $R(t_1,\dots,t_p,x_{i_1},\dots,x_{i_{d-p}})$. Here, $x_{i_1},\dots,x_{i_{d-p}}$ is the complete list of distinct variables among $\tu x$ that do not occur in $t_1,\dots,t_p$.
\end{proof}

% remark
\begin{remark}\label{rq:sans repetition}
The proof of Fact~\ref{thm:sans repetition} allows enhancing its statement: each atomic subformula of $\tilde{\phi}$ that involves an input symbol $Q_a$, for some $a\in\Sg$, has the form $Q_a(x_{\al_{1}},\dots,x_{\al_{d-1}})$ where $\al$ is an injection from $\{1,\dots,d-1\}$ into $\{1,\dots,d\}$.\margenote{\`A v\'erifier} 
\end{remark}

% fact
\begin{fact}\label{thm:un seul succ ou permutation}
On $\coordstruc{d-1}$, any formula $\Phi=\ex\tu R\fa\tu x\phi(\tu x,\tu R,\sg)\in\esorel{}(\fa^{d},\arity d)$ of signature $\sg$ can be written in such a way that each atomic subformula over $\tu R$ of $\phi$ has one of the two forms:
$R(\tu x^{(i)})$ or $R(\tu x_{\pi})$, where $\pi\in\perm(d)$.
\end{fact}

% proof
\begin{proof}
We prove the result for $d=3$. The general case is similar. %First, assume \wlg that $\Phi$ is under conjunctive normal form. 
Let $\ell$ be the maximal value of an $i\in\N$ such that $\suc^{i}(x)$ occurs in $\Phi$, for any $x\in\tu x$. For each $R\in\tu R$, we introduce new $d$-ary relation symbols $R_{i,j,k}$ for every $i,j,k\leq\ell$.  
% and consider some atomic formula $R(\suc^{p}u_{1},\suc^{q}u_{2},\suc^{r}u_{3})$ occurring in $\Phi$, where each $u_i$ belongs to $\tu x$, $R\in\sg\cup\{\tu R\}$ \margenote{$R\notin\sg$ pour des questions d'arit\'e}
%and  $p+q+r\neq 0$. (If such a formula does not exist, we are done.) Let's introduce new relation symbols $R_{i,j,k}$, for every $i\leq p$, $j\leq q$, $k\leq r$. 
We want to force the following interpretations of the $R_{i,j,k}$'s: \[R_{i,j,k}(u_1,u_2,u_3)=R(\suc^{i}u_{1},\suc^{j}u_{2},\suc^{k}u_{3}).\] 
This is done inductively, with the formulas: 
\begin{itemize}
\item  ${\dst \fa\tu x: R_{0,0,0}(x_1,x_2,x_3)\ssi R(x_1,x_2,x_3)}$
\item  ${\dst \fa\tu x: \Et_{i<\ell}\,\Et_{j,k\leq\ell}( R_{i+1,j,k}(x_1,x_2,x_3)\ssi R_{i,j,k}(\suc(x_1),x_2,x_3) )}$
\item  ${\dst \fa\tu x: \Et_{j<\ell}\,\Et_{i,k\leq\ell}( R_{i,j+1,k}(x_1,x_2,x_3)\ssi R_{i,j,k}(x_1,\suc(x_2),x_3) )}$
\item  ${\dst \fa\tu x: \Et_{k<\ell}\,\Et_{i,j\leq\ell}( R_{i,j,k+1}(x_1,x_2,x_3)\ssi R_{i,j,k}(x_1,x_2,\suc(x_3)) )}$
\end{itemize}
Factorizing the quantifications and using notations of Definition~\ref{def:tuples et permutations}, the conjunction of these formulas can be written:
\[
\fa\tu x\left\{\begin{array}{ll}
{ R_{0,0,0}(\tu x)\ssi R(\tu x)}  \et \\ 
{\dst \Et_{i<\ell}\,\Et_{j,k\leq\ell}( R_{i+1,j,k}(\tu x)\ssi R_{i,j,k}(\tu x^{(1)}) )}  \et \\
{\dst \Et_{j<\ell}\,\Et_{i,k\leq\ell}( R_{i,j+1,k}(\tu x)\ssi R_{i,j,k}(\tu x^{(2)}) )}  \et \\
{\dst \Et_{k<\ell}\,\Et_{i,j\leq\ell}( R_{i,j,k+1}(\tu x)\ssi R_{i,j,k}(\tu x^{(3)}) )} 
\end{array}\right\}
\]
Let us denotes by $\decomp(R,(R_{i,j,k})_{i,j,k\leq\ell})$ this last formula. It clearly fulfills the condition of the statement. Now, consider the formula 
% equation
\begin{equation}\label{eq:decomp et tilde(phi)}
\ex\tu R\ex ((R_{i,j,k})_{i,j,k\leq\ell}))_{R\in\tu R} : 
\Et_{R\in\tu R}\decomp(R,(R_{i,j,k})_{i,j,k\leq\ell})\ \et\ 
\fa\tu x\Tilde{\phi},
\end{equation}
where $\Tilde{\phi}$ is obtained from $\phi$ by the substitutions $$R(\suc^{i}x_{\al_1},\suc^{j}x_{\al_2},\suc^{k}x_{\al_3})\devient R_{i,j,k}(x_{\al_1},x_{\al_2},x_{\al_3}).$$ 
Then, the formula~\refeq{eq:decomp et tilde(phi)} is equivalent to $\Phi$ and also fits the requirements of Fact~\ref{thm:un seul succ ou permutation}. It is the rewriting of $\Phi$ announced.
\end{proof}

As a result of Fact~\ref{thm:sans repetition}, Remark~\ref{rq:sans repetition} and Fact~\ref{thm:un seul succ ou permutation}, each $\esorel{}(\fa^{d},\arity d)$-formula of signature $\sg$ has a conjonctive normal form of the shape: 
\[\Phi\equiv\ex\tu R\fa\tu x\Et\Ou\pm\left\{
\begin{array}{l}
\min(\suc^{i}(x)),\ \max(\suc^{i}(x)), \\ \suc^{i}(x)=\suc^{j}(y), \\ Q_a(\tu x_{\iota}),\ R(\tu x_{\beta}),\ R(\tu x^{(i)})
\end{array}\right\}
\]
Furthermore, the trick used in the proof of Fact~\ref{thm:un seul succ ou permutation} also allows writing atoms involving $\min$, $\max$ or equalities, under the form $\max(x)$, $\min(x)$ 
and $x=y$ for some $x,y\in\tu x$. 

%Finally, we can state: 

%% synthese
%\synthese{
%On pictures, any formula $\Phi\in\esorel{}(\fa^{d},\arity d)$ can be written:
%% equation
%\begin{equation}\label{eq:un seul succ ou permutation}
%\Phi\equiv\ex\tu R\fa\tu x\Et\Ou\pm\left\{
%\begin{array}{l}
%\min(x),\ \max(x),\ x=y, \\ Q_a(\tu x_{\iota}),\ R(\tu x_{\beta}),\ R(\tu x^{(i)})
%\end{array}\right\}
%\end{equation}
%where $a\in\Sg$, $R\in\tu R$, $x,y\in\tu x$, $\iota\in\inject(d)$, $\beta\in\perm(d)$ and $i\in\{1,\dots,d\}$. 
%}

It remains to prove that we can get rid of the atomic formulas $R(\tu x_{\beta})$, where $\beta\neq\id$. This part is rather technical, so we provide some preliminary explanations before stating the logical framework which allows the normalization. In order to get rid of each literal of the form $R(x_{\beta})$, we will divide the set $R\subseteq\dom n^d$ in $d!$ relations $R_{\beta}\subseteq\dom n^d$, each corresponding to a given permutation $\beta$ of $\{1,\dots,d\}$.

% definition
\begin{definition}\label{def:R_alpha}
For $R\subseteq\dom n^d$ and for each $\al\in\perm(d)$, we define a $d$-ary relation \textdef{$R_{\al}$} on $\dom n$ by: $R_{\al}=\{\tu x\in\sby{\id}\st R(\tu x_{\al^{-1}})\}.$ Alternatively, $R_{\al}$ can be defined by: 
$R_{\al}=\{\tu x_{\al}: \tu x\in R\cap\sby{\al} \}$.
\end{definition}

Thus, Definition~\ref{def:R_alpha} associates with each $R\subseteq\dom n^d$ a family $(R_{\al})_{\al\in\perm(d)}$ of relations, each of which is entirely contained in the set $\sby{\id}$. 
This family is intended to represent $R$ through its $d!$ fragments according to the partition $\dom n^d=\bigcup_{\al\in\perm(d)}\sby{\al}$. Namely, each $R_{\al}$ encodes the fragment $R\cap \sby{\al}$ over $\sby{\id}$. 

\medskip

Actually, $\bigcup_{\al\in\perm(d)}\sby{\al}$ is not really a partition, since the $\sby{\al}$'s can overlap. Hence, Definition~\ref{def:R_alpha} induces some connexions between the relations $R_{\al}$: if some $\tu x$ is both in $\sby{\al}$ and in $\sby{\beta}$, or equivalently, if $\tu x_{\al}=\tu x_{\beta}$, then $\tu x\in R\cap\sby{\al}$ iff $\tu x\in R\cap\sby{\beta}$ and hence, by Definition~\ref{def:R_alpha}: $R_{\al}(\tu x_{\al})=R_{\beta}(\tu x_{\beta})$. We will keep in mind : 
% equation
\begin{equation}\label{eq:connexions des R_a}
\fa\al,\beta\in\perm(d),\fa\tu x\in\dom n^d: \tu x_{\al}=\tu x_{\beta}\Imp R_{\al}(\tu x_{\al})=R_{\beta}(\tu x_{\beta}).
\end{equation}

The following lemma states that condition~\refeq{eq:connexions des R_a} ensures that the $R_{\al}$'s issue from a single relation $R$, according to Definition~\ref{def:R_alpha}. Besides, a new formulation of the condition is given in Item~3 of the lemma, that will better fit our syntactical restrictions. 

% lemma
\begin{lemma}\label{thm:(R_a) code R}
Let $(R_{\al})_{\al\in\perm(d)}$ be a family of $d$-ary relations on $\dom n$ such that $R_{\al}\subseteq\sby{\id}$ for each $\al$. The following are equivalent:
\begin{enumerate}
\item  $\ex R\subseteq\dom n^d$ such that $R_{\al}=\{\tu x\in\sby{\id}\st R(\tu x_{\al^{-1}})\}$ for each $\al\in\perm(d)$ ;
\item  $\fa\al,\beta\in\perm(d)$, $\fa\tu x\in\dom n^d$: $\tu x_{\al}=\tu x_{\beta}\Imp R_{\al}(\tu x_{\al})=R_{\beta}(\tu x_{\beta})$ ; 
%\item  $\fa\al,\beta\in\perm(d)$, $\fa\tu x\in\dom n^d$: $\tu x=\tu x_{\beta\al^{-1}}\Imp R_{\al}(\tu x)=R_{\beta}(\tu x)$ ; 
\item  $\fa\al\in\perm(d)$, $\fa\tau\in\Transpo(d)$, $\fa\tu x\in\dom n^d$: $\tu x=\tu x_{\tau}\Imp R_{\al}(\tu x)=R_{\al\tau}(\tu x)$.

(Recall $\Transpo(d)$ denotes the set of transpositions over $\{1,\dots,d\}$.)
\end{enumerate}
\end{lemma}

% proof
\begin{proof}
{$1\Imp 2$: } See Equation~\refeq{eq:connexions des R_a}.

\medskip {$2\Imp 1$: } 
For $(R_{\al})_{\al\in\perm(d)}$ fulfilling $2$, consider the relation $R\subseteq\dom n^d$ defined by: 
% equation
\begin{equation}\label{eq:def R} 
R(\tu x) \IFF R_{\al}(\tu x_{\al})\text{ for some $\al$ such that }\tu x_{\al}\up. 
\end{equation} 
This definition is well formed, since for any $\al,\beta\in\perm(d)$ and any $\tu x\in\dom n^d$ such that both $\tu x_{\al}\up$ and $\tu x_{\beta}\up$ hold, we have $\tu x_{\al}=\tu x_{\beta}$ and thus, by $2$, $R_{\al}(\tu x_{\al})=R_{\beta}(\tu x_{\beta})$. 
Now, let $\al\in\perm(d)$. For any $\tu x\in\sby{\id}$ we have $(\tu x_{\al^{-1}})_{\al}\up$ (since $(\tu x_{\al^{-1}})_{\al}=\tu x$) and hence, by~\refeq{eq:def R}, $R_{\al}(\tu x_{})=R_{\al}((\tu x_{\al^{-1}})_{\al})=R(\tu x_{\al^{-1}})$. Besides, $R_{\al}(\tu x_{})=\faux$ for any $\tu x\notin\sby{\id}$, since $R_{\al}\subseteq\sby{\id}$. Thus $R_{\al}$ is obtained from $R$ as required in $1$. 

\medskip {$2\Imp 3$: } 
Let $\al\in\perm(d)$, $\tau\in\Transpo(d)$ and $\tu x\in\dom n^d$ such that $\tu x=\tu x_{\tau}$. Set $\tu y=\tu x_{\al^{-1}}$. Then, $\tu y_{\al}=\tu x=\tu x_{\tau}=(\tu y_{\al})_{\tau}=\tu y_{\al\tau}$. Therefore we get by~2: $R_{\al}(\tu y_{\al})=R_{\al\tau}(\tu y_{\al\tau})$, and hence: $R_{\al}(\tu x)=R_{\al\tau}(\tu x)$.

\medskip {$3\Imp 2$: } 
Let $\al,\beta\in\perm(d)$ and $\tu x\in\dom n^d$ such that $\tu x_{\al}=\tu x_{\beta}$. For $\tu y=\tu x_{\al}$, the equality $\tu x_{\al}=\tu x_{\beta}$ can be written $\tu y=\tu y_{\al^{-1}\beta}$. It means that the permutation $\al^{-1}\beta$ exchanges integers that index equal components of $\tu y$. It is easily seen that this property can be required for each transposition occuring in a decomposition of $\al^{-1}\beta$ on $\Transpo(d)$. That is, there exist some transpositions $\tau_1,\dots,\tau_k\in\Transpo(d)$ such that $\al^{-1}\beta=\tau_1\dots\tau_k$ and $\tu y=\tu y_{\tau_1}=\tu y_{\tau_1\tau_2}=\dots=\tu y_{\tau_1\dots\tau_k}$. Then, applying $3$ to these successive tuples, we get: 
$R_{\al}(\tu y)=R_{\al\tau_1}(\tu y_{\tau_1})=R_{\al\tau_1\tau_2}(\tu y_{\tau_1\tau_2})=\dots=R_{\al\tau_1\dots\tau_k}(\tu y_{\tau_1\dots\tau_k})$. 
Hence $R_{\al}(\tu y)=R_{\beta}(\tu y_{\al^{-1}\beta})$, that is $R_{\al}(\tu x_{\al})=R_{\beta}(\tu x_{\beta})$, as required. 
\end{proof}

\begin{lemma}\label{thm:R(x_b) and R(x^i) vs R(x_a)}
Let $R$ and $(R_{\al})_{\al\in\perm(d)}$ be defined as in Definition~\ref{def:R_alpha}. Let $\al,\beta\in\perm(d)$ and $i\in\{1,\dots,d\}$. For any $\tu x\in\sby{\al}$: 
\begin{enumerate}
\item  $R(\tu x_{\beta})$ is equivalent to $R_{\beta^{-1}\al}(\tu x_{\al}).$ 
\item  $R(\tu x^{(i)})$ is equivalent to: 
\[
\begin{array}{c}
\left( x_{i}=x_{\al_{d}} \et R_{\al\transpo{\al^{-1}(i)}{d}}((x_{\al})^{(d)})  \right)\ \ou %\\
{\dst \Ou_{i\leq k<d}} \left( x_{i}=x_{\al_{k}}<x_{\al_{k+1}}\et R_{\al\transpo{\al^{-1}(i)}{k}}((x_{\al})^{(k)}) \right).
\end{array}
\]
\end{enumerate}
\end{lemma}

% proof
\begin{proof}
{\bf 1.} \ 
We have seen that $R(\tu x)$ holds iff $R_{\gamma}(\tu x_{\gamma})$ holds for \emph{any} $\gamma\in\perm(d)$ such that $\tu x_{\gamma}\up$. Thus we  get, for a given $\beta\in\perm(d)$: 
$R(\tu x_{\beta})$ iff $R_{\gamma}((\tu x_{\beta})_{\gamma})$ holds for any $\gamma\in\perm(d)$ such that $(\tu x_{\beta})_{\gamma}\up$. That is, since $(\tu x_{\beta})_{\gamma}=\tu x_{\beta\gamma}$: 
% equation
\begin{equation}\label{eq:R(x_b) ssi R_g(x_gb)}
R(\tu x_{\beta}) \IFF R_{\gamma}(\tu x_{\beta\gamma})\text{ holds for any $\gamma\in\perm(d)$ such that $\tu x_{\beta\gamma}\up$}.
\end{equation}
In particular, $\tu x_{\beta\gamma}\up$ iff $\tu x_{\beta\gamma}=\tu x_{\al}$, since $\tu x\in\sby{\al}$. Thus, $\gamma=\beta^{-1}\al$ is one of the permutations such that $\tu x_{\beta\gamma}\up$. Thus, replacing $\gamma$ by $\beta^{-1}\al$ in Equation~\refeq{eq:R(x_b) ssi R_g(x_gb)}, we get the sought result.  

{\bf 2.} \ From $\tu x^{(i)}=(x_{1},\dots,\suc(x_{i}),\dots,x_{d})$ we get: 
\[(\tu x^{(i)})_{\al}=(x_{\al_1},\dots,x_{\al_{j-1}},\ \suc(x_{i}),\ x_{\al_{j+1}},\dots,x_{\al_d})\] 
\margenote{Revoir!}
where $j=\al^{-1}(i)$. Since $x_{\al_1}\leq\dots\leq x_{\al_d}$, the above tuple $(\tu x^{(i)})_{\al}$ is almost increasingly ordered. More precisely, there exists $k\in\{1,\dots,d\}$ such that: 
$$x_{\al_1}\leq\dots\leq x_{\al_{j-1}}\leq x_{i}=x_{\al_{j+1}}=\dots=x_{\al_{k}}\leq x_{\al_{k+1}}\leq\dots\leq x_{\al_d},$$
where $j=\al^{-1}(i)$. Clearly, the largest such $k$ is characterized by: $(k=d)\OR (k<d\AND x_{i}=x_{\al_{k}}<x_{\al_{k+1}})$. Or equivalently, by: 
% equation
\begin{equation}\label{eq:carac k} 
(x_{i}=x_{\al_{d}})\OR (k<d\AND x_{i}=x_{\al_{k}}<x_{\al_{k+1}}).
\end{equation}
If we denote by $\transpo{j}{k}$ the transposition over $\{1,\dots,d\}$ which permutes $j$ and $k$, the definition of $k$ yields that the tuple 
\[
\begin{array}{c}
(x^{(i)})_{\al\transpo{j}{k}}= %\\ 
(x_{\al_1},\dots,x_{\al_{j-1}},\text{\framebox{$x_{\al_{k}}$}},x_{\al_{j+1}},\dots,x_{\al_{k-1}},\text{\framebox{$\suc(x_{i})$}},x_{\al_{k+1}},\dots,x_{\al_d})
\end{array}
\] 
is non decreasing. Hence, $R(\tu x^{(i)})=R_{\al\transpo{j}{k}}((x^{(i)})_{\al\transpo{j}{k}})$. 
Besides, since $x_{\al_{k}}=x_{i}$, the tuple $(x^{(i)})_{\al\transpo{j}{k}}$ above can also be written: 
\[
\begin{array}{c}
(x^{(i)})_{\al\transpo{j}{k}}= %\\ 
(x_{\al_1},\dots,x_{\al_{j-1}},\text{\framebox{$x_{i}$}},x_{\al_{j+1}},\dots,x_{\al_{k-1}},\text{\framebox{$\suc(x_{\al_{k}})$}},x_{\al_{k+1}},\dots,x_{\al_d}).
\end{array}
\] 
That is: $(x^{(i)})_{\al\transpo{j}{k}}=(x_{\al})^{(k)}$. Therefore: $R(\tu x^{(i)})=R_{\al\transpo{j}{k}}((x_{\al})^{(k)})$. 
Reminding that $j=\al^{-1}(i)$, we can finally state: there exists a sole $k\in\{i,\dots,d\}$ defined by~\refeq{eq:carac k}, and for this $k$ we have: \ 
$R(\tu x^{(i)})=R_{\al\transpo{\al^{-1}(i)}{k}}((x_{\al})^{(k)})$. The conclusion easily proceeds.
\end{proof}

% proposition
\begin{proposition}\label{thm:normal-sort-R} 
For $d>1$, $\esorel{}(\fa^{d},\arity d)\subseteq\esorel{}(\fa^{d},\arity d,\semisorted)$ on $\coordstruc{d-1}$.
\end{proposition}

% proof
\begin{proof}
To simplify, assume we want to translate in $\esorel{}(\fa^{d},\arity d,\semisorted)$ some $\esorel{}(\fa^{d},\arity d)$-formula of the very simple shape: $\Phi\equiv\ex R\fa\tu x\phi(\tu x,R)$, where $R$ is a (single) $d$-ary relation symbol, $\tu x$ is a $d$-tuple of first-order variables, and $\phi$ is a quantifier-free formula. Since the sets $\sby{\al}$, $\al\in\perm(d)$, cover the domain $\dom n$, we obtain an equivalent rewriting of $\Phi$ with the following artificial relativization: $\Phi\equiv\ex R\fa\tu x{\Et_{\al\in\perm(d)}}\left( \tu x\in\sby{\al}\imp\phi \right)$. Furthermore, all atomic subformulas of $\phi$ built on $R$ can be assumed of the form $R(\tu x_{\beta})$ or $R(\tu x^{(i)})$, thanks to Fact~\ref{thm:un seul succ ou permutation}.

To get rid of these literals, we substitute to $R$ a tuple of relations $(R_{\al})_{\al\in\perm(d)}$ that encode $R$ on the sets $\sby{\al}$. Recall we proved in Lemma~\ref{thm:(R_a) code R} that this substitution is legal as soon as $R_{\al}\subseteq\sby{\id}$ and $R_{\al}(\tu x)=R_{\al\tau}(\tu x)$ for all $\al\in\perm(d)$, $\tau\in\Transpo(d)$ and every $\tu x\in\dom n^d$ such that $\tu x_{\tau}=\tu x$. Then, Lemma~\ref{thm:R(x_b) and R(x^i) vs R(x_a)} gives the translation of $R$-atomic formulas into formulas expressed in term of the $R_{\al}$'s. All in all, we get the equivalence of the initial formula $\Phi$ to the following: 

% equation
\begin{equation}\label{eq:normaliser ex R fa x phi}
\begin{array}{c}
\ex (R_{\al})_{\al\in\perm(d)} %\\
\left\{\begin{array}{ll}
\fa\tu x {\dst \Et_{\al\in\perm(d)}\left(  R_{\al}(\tu x)\imp \tu x\in\sby{\id} \right) }  \et \\ 
\fa\tu x {\dst \Et_{\al\in\perm(d)}\Et_{\tau\in\Transpo(d)}\left( \tu x_{\tau}=\tu x\imp (R_{\al}(\tu x)\ssi R_{\al\tau}(\tu x))\right) }  \et \\ 
\fa\tu x {\dst \Et_{\al\in\perm(d)}\left( \tu x\in\sby{\al}\imp\phi_{\alpha}(\tu x,(R_{\gamma})_{\gamma\in\perm(d)})\right) }
\end{array}\right\}
\end{array}
\end{equation}
where each $\phi_{\alpha}$ is obtained from $\phi$ by the substitutions: 
\begin{itemize}
\item  $R(\tu x_{\beta})\devient R_{\beta^{-1}\al}(x_{\al})$
\item  $R(\tu x^{(i)})\devient
\left\{\begin{array}{l}
\left( x_{i}=x_{\al_{d}} \et R_{\al\transpo{\al^{-1}(i)}{d}}((x_{\al})^{(d)})  \right)\ \ou \\ 
{\dst \Ou_{i\leq k<d}} \left( x_{i}=x_{\al_{k}}<x_{\al_{k+1}}\et R_{\al\transpo{\al^{-1}(i)}{k}}((x_{\al})^{(k)}) \right)
\end{array}\right\}$
\end{itemize}
%replacing every subformulas $R(\tu x_{\beta})$ of $\phi$ by $R_{\beta^{-1}\al}(\tu x_{\al})$. 
The first two conjuncts of~\refeq{eq:normaliser ex R fa x phi} ensure that the family $(R_{\al})_{\al\in\perm(d)}$ encodes a relation $R$ (see Lemma~\ref{thm:(R_a) code R}) ; the third interprets assertions of the form $R(\tu x_{\beta})$ and $R(\tu x^{(i)})$ according to the modalities described in Lemma~\ref{thm:R(x_b) and R(x^i) vs R(x_a)}. Because of permutability of the conjunction and the universal quantifier, this third conjunct can be rewritten: 
% equation
\begin{equation}\label{eq: x in a -> phi_a}
{\dst \Et_{\al\in\perm(d)}}\fa\tu x:\tu x\in\sby{\al}\imp\phi_{\alpha}(\tu x,(R_{\gamma})_{\gamma\in\perm(d)})
\end{equation}

For a fixed conjunct in~\refeq{eq: x in a -> phi_a}, \ie for a fixed $\alpha$, all atomic subformulas of $\phi_{\alpha}$ built on the $R_{\gamma}$'s have the form $R_{\gamma}(\tu x_{\al})$ or $R_{\gamma}((\tu x_{\al})^{(i)})$ for some $\gamma\in\perm(d)$ and some $i\in\{1,\dots,d\}$. Hence, the substitution of variables $\tu x/\tu x_{\al}$ allows to write such a conjunct as: 
$\fa\tu x:\tu x\in\sby{\id}\imp\tilde\phi_{\alpha}$ where $\tilde\phi_{\alpha}\equiv\phi_{\alpha}(\tu x/\tu x_{\al})$ only involves $(R_{\gamma})$-subformulas of the form $R_{\gamma}(\tu x)$ or $R_{\gamma}(\tu x^{(k)})$ for some $\gamma\in\perm(d)$ and $k\in\{1,\dots,d\}$. 
Finally, the initial formula $\Phi$ is proved equivalent to: 
\[
\begin{array}{c}
\ex (R_{\al})_{\al\in\perm(d)}%: \\
\left\{\begin{array}{ll}
\fa\tu x {\dst \Et_{\al\in\perm(d)}\left(  R_{\al}(\tu x)\imp \tu x\in\sby{\id} \right) }  \et \\ 
\fa\tu x {\dst \Et_{\al\in\perm(d)}\Et_{\tau\in\Transpo(d)}\left( \tu x_{\tau}=\tu x\imp (R_{\al}(\tu x)\ssi R_{\al\tau}(\tu x))\right) }  \et \\ 
\fa\tu x {\dst \Et_{\al\in\perm(d)}\left( \tu x\in\sby{\id}\imp\tilde\phi_{\alpha}(\tu x,(R_{\gamma})_{\gamma\in\perm(d)})\right) }
\end{array}\right\}
\end{array}
\]
that fulfills the requirement of Proposition~\ref{thm:normal-sort-R}.
\end{proof}

%As a result of Proposition~\ref{thm:normal-sort-R}, we will remember:
%% synthese
%\synthese{
%Each $\esorel{}(\fa^{d},\arity d)$-formula of signature $\sg$ can be writen, on pictures: 
%% equation
%\begin{equation}\label{eq:esorel(d)=esorel(d,)}
%\Phi\equiv\ex\tu R\fa\tu x\Et\Ou\pm\left\{
%\begin{array}{l}
%\min(x),\ \max(x),\ x=y,\ x<y \\ 
%Q_a(\tu x_{\iota}),\ R(\tu x),\ R(\tu x^{(i)})
%\end{array}\right\}
%\end{equation}
%where $a\in\Sg$, $R\in\tu R$, $x,y\in\tu x$, $\iota\in\inject(d)$ and $i\in\{1,\dots,d\}$. 
%}
%\margenote{$\suc$ ?}

% subsection
\subsection{Getting rid of arithmetic}\label{sec:arithmetic}	

Does mean introducing new second order variables of arity $2$, we can assume that our formula $\Phi\in\esosig(\fa^{d},\arity d)$ -- where each relation symbol of arity $d \leq 2$ occurs in normalised (\ie sorted) form -- involves no comparison (equality, inequality) relation. We obtain that in two successive steps.

\medskip

First, if $\Phi$ involves inequalities $<$ and $>$, then it is equivalent to the following formula $\Phi'$ without inequalities (but with equalities) and two new binary relation symbols $\infe$ and $\supe$:
\[
\begin{array}{c}
\Phi'\equiv\ex\infe\ \ex\supe : \tilde\Phi\ \et \fa x_1,x_2 %\\ 
\left\{\begin{array}{rclll}
( x_1=x_2 \ou x_1\infe x_2 ) & \imp & ( \lnot\max(x_2) \imp x_1\infe \suc(x_2) )  & \et \\
( x_1=x_2 \ou x_1\supe x_2 ) & \imp & ( \lnot\max(x_1) \imp \suc(x_1)\supe x_2 )  & \et  \\
x_1\infe x_2 & \imp & (\lnot (x_1\supe x_2)\et\lnot(x_1=x_2) )                    & \et  \\
x_1\supe x_2 & \imp & \lnot(x_1=x_2) 
\end{array}\right\}
\end{array}
\]
where $\tilde\Phi$ is obtained from $\Phi$ by the substitutions $u < v\devient u\infe v$ (resp. $u > v\devient u\supe v$).  This is justified as follows: the first two conjuncts of the subformula $\fa x_1\fa x_2\{\cdots\}$ express that $x_1<x_2\Imp x_1\infe x_2$ (resp. $x_1 > x_2\Imp x_1\supe x_2$). The third and fourth conjuncts express that the three relations $\infe, \supe \textrm{ and } =$ are totally disjoint. That implies that $\infe$ and $\supe$ have their exact meaning. Notice also that the occurrences of $\infe$ and $\supe$ preserve the sorted property. 
 
\medskip

Secondly, if $\Phi$ involves equalities (without any inequality), it is equivalent to the following formula $\Phi'$, written  without the symbol "$=$" but with the new binary symbol $\egal$:
%\[
%\begin{array}{c}
%\Phi'\equiv\ex\egal:\ \tilde\Phi\ \et\ \fa x_1,x_2 \\
%\left\{\begin{array}{rcll}
%\min(x_1)  \imp  (\min(x_2)\ssi x_1\egal x_2)    \et \\
%\min(x_2)  \imp  (\min(x_1) \ssi x_1\egal x_2)  \et \\
%(\lnot \max(x_1) \et \lnot \max(x_2))  \imp  (x_1\egal x_2\ssi\suc(x_1)\egal\suc(x_2))
%\end{array}\right\}
%\end{array}
%\]
$$\Phi'\equiv\ex\egal:\ \tilde\Phi\ \et\ \fa x_1,x_2 \Psi$$
where $\tilde\Phi$ is obtained from $\Phi$ by replacing each equality $u=v$ by $u\egal v$, and $\Psi$ is the conjunction of the following formulas: 
\begin{itemize}
\item  $\min(x_1)  \imp  (\min(x_2)\ssi x_1\egal x_2)$;
\item  $\min(x_2)  \imp  (\min(x_1) \ssi x_1\egal x_2)$;
\item  $(\lnot \max(x_1) \et \lnot \max(x_2))  \imp  (x_1\egal x_2\ssi\suc(x_1)\egal\suc(x_2))$.
\end{itemize}
Notice that this transformation preserves sorted property. The results obtained until now can be recapitulated as follows:

%\margenote{attention: $\infe$, $\supe$, $\egal$ ne sont pas d'arite d-1}
% synthese
\synthese{
On $(d-1)$-pictures, each $\esosig(\fa^{d},\arity d)$-formula can be written under the form: 
% equation
\begin{equation}\label{eq:elim < et =}\Phi\equiv\ex\tu R\fa\tu x\Et\Ou\pm \left\{ \begin{array}{l} \min(x),\ \max(x), \\ Q(\tu x_{\iota}),\ R(\tu x),\ R(\tu x^{(i)}) \end{array} \right\}\end{equation}
Here, $\tu R$ (resp. $\tu Q$) is a list of relation symbols of arity $d$ (resp. $d-1$), $x\in\tu x$, $Q\in\cup (Q_a)_{a\in\Sg}$, $R\in\tu R$, $\iota\in\inject(d)$, $i\in\{1,\dots,d\}$.
}% fin synthese

%\section
\section{Normalization of input relations}\label{sec:sort Q}
%	\input normalSorted.tex
%% sortQ.tex

\providecommand{\flip}[1]{\text{flip}(#1)}
\newcommand{\graphe}{\cl G}
\newcommand{\pa}[2]{[#2]}%{\cdot #1#2}%{[#1#2]}%
\newcommand{\px}[2]{\cdot #1#2}%{[#1#2]}%
\newcommand{\truc}[1]{[#1]}
\newcommand{\axiom}{\textsf{A}}
\newcommand{\refaxiom}[1]{$(\axiom_{\ref{#1}})$}
\newcommand{\axiomlog}{\textsf{F}}
\newcommand{\refaxiomlog}[1]{$(\axiomlog_{\ref{#1}})$}
\newcommand{\haut}{\uparrow}
\newcommand{\bas}{\downarrow}
\newcommand{\gche}{\leftarrow}
\newcommand{\dte}{\rightarrow}
\newcommand{\ul}[1]{\underline{#1}}

% subsection\subsection{Normalization of input relations}\label{sec:sort Q}

%% Elimination des Q

Let's come back to the normalization of the formula $\Phi$ of Equation~\refeq{eq:elim < et =}. For simplicity, assume that it has the restricted form $\Phi\equiv\ex \tu R \fa \tu x \phi(\tu x, \tu R,Q,\sigma)$, with only one relation symbol $Q$ of arity $d-1$. 

\medskip

We aim at defining a tuple of relations $(Q_{\al})_{\al\in\inject(d)}$, in such a way that $Q_\al(\tu x)=Q(\tu x_\al)$ for each $\tu x$. Clearly, such relations will allow to write $\Phi$ under the desired form, by replacing each subformula $Q(\tu x_\al)$ by the sorted translation: $Q_\al(\tu x)$. The difficulty is to express this definition with our syntactical restrictions, that is, without involving any $\tu x_{\al}$ with $\al\neq\id$. %Of course, the simple transcription of the definition, $\fa\tu x: Q_\al(\tu x)\ssi Q(\tu x_\al)$ is  
\medskip

Notice that the strategy used in Section~\ref{sec:sort R} to "sort" atomic subformulas $R(\tu x_{\al})$ build on any \emph{existentially quantified} $d$-ary relation $R$ is no more avalaible, since it means  suppressing $R$ in favour of some new existentially quantified relations. Of course, we can't operate like this with the \emph{input} relation $Q$. 
%We will associate with $Q$ a family of "temporary" $d$-ary relations $(T_\al)_{\al \in \inject(d)}$ defined inductively. Intuitively, given $\al\in\inject(d)$ and the corresponding $\apf\al\notin\im(\al)$, the only significant tuples $\tu x$ in $T_\al$ are those for which $x_{\apf\al}=0$. Hence, the "column" of rank $\apf\al$ in $T_\al$ can be used as a kind of buffer. The $(T_\al)$'s are defined by induction on the prefix order of the decomposition tree of $\inject(d)$. 
To give an hint of the method used in this section, let us consider an easy example. 

\paragraph{An easy case.} Consider the case where $d=2$. We deal with two first-order variables $x$ and $y$ and we only accept atoms of the form $Q(x)$, $R(x,y)$, $R(\suc(x),y)$ and $R(x,\suc(y))$ for any input unary relation $Q$ and any guessed binary relation $R$. How can we tackle occurences of some atom $Q(y)$ in the formula ? A natural idea is to define a new binary relation $Q_2$ in such a way that $Q_2(x,y)=Q(y)$ holds for any $x,y$. (We denote it $Q_{2}$ to refer both to $Q$ and to the projection of $(x,y)$ on its \emph{second} component.) Hence, we set: 
\[%\begin{equation}\label{eq:Q2}
Q_2=\{(x,y) : Q(y)\}.
\]%\end{equation}
Thus, any atom $Q(y)$ could be replaced by the sorted atom $Q_2(x,y)$. But the logical definition of $Q_2$ with our syntactical constraints compells to introduce an additive binary relation $T$ that will be used as a buffer to \emph{transport} the information $Q(y)$ into the expression $Q_2(x,y)$. We set 
\[%\begin{equation}\label{eq:Q1}
T=\{(x,y) : Q(x+y)\}.
\]%\end{equation}
Clearly, $T$ is inductively defined from $Q$ by the assertions $T(x,0)=Q(x)$ and $T(x+1,y)=T(x,y+1)$. Besides, $Q_2$ is defined from $T$ by $Q_2(0,y)=T(0,y)$ and $Q_2(x,y)=Q_2(x+1,y)$. All these assertions can be rephrazed in our logical framework, with the following formulas: 
\[
\begin{array}{c}
\fa x,y %\\
\left\{\begin{array}{rclr}
\min(y) & \imp & \left(T(x,y)\ssi Q(x)\right) \hspace{1ex} \et \\
\left(\neg\max(x)\et\neg\max(y)\right) & \imp & T(\suc(x),y)\ssi T(x,\suc(y))  
\end{array}\right\}
\end{array}
\]
\[
\fa x,y 
\left\{
\begin{array}{l}
\min(x)\imp \left(Q_2(x,y)\ssi T(x,y)\right) \hspace{1ex} \et \\ 
Q_2(x,y)\ssi Q_2(\suc(x),y)
\end{array}\right\}
\]
Now, it remains to insert this defining formulas in the inital formula $\Phi$ to be normalized, and to replace each occurence of $Q(y)$ by $Q_2(x,y)$. Of course, such a construction has to be carried on for each input unary relation $Q$. 

% paragraph
\paragraph{The general case.}

Given $i,j\in\{1,\dots,d\}$ and $\al\in\perm(d)$, we denote by $(ij)$ the transposition that exchanges $i$ and $j$, and by $\al(ij)$ the composition of $\al$ and $(ij)$. It is well-known that each permutation $\al$ can be written as a product of transpositions, $\al=(u_1v_1)(u_2v_2)\dots(u_pv_p)$. It is easily seen that this product can be chosen in such a way that $v_i=u_{i+1}$ for any $i$. This is because if some sequence $(ab)(cd)$ with $b\neq c$ occurs, it can be replaced by $(ca)(ab)(bc)(cd)$, and a well chosen iteration of such rewritings yields the desired decomposition. This can be further refined, by fixing at $d$ one element of the first transposition involved in the decomposition and by prohibiting useless sequence as $(ab)(ba)$. Finally, each $\al\in\perm(d)$ can be written 
% equation
\begin{equation}\label{eq:a=du1...uk} 
\al=(du_1)(u_1u_2)\dots(u_{k-2}u_{k-1})(u_{k-1}u_{k}), 
\end{equation}
where $u_{i}$, $u_{i+1}$ and $u_{i+2}$ are pairwise distinct elements of $\{1,\dots,d\}$ for any $i$. 
We call \textdef{alternated factorization of $\al$} such a decomposition. 
%
%Because of the way it is involved in this section, we introduce a special notation for alternated factorizations: we set $$\textdef{$[du_1\dots u_k]$}\egaldef (du_1)(u_1u_2)\dots(u_{k-2}u_{k-1})(u_{k-1}u_{k})$$ 
%when $u_{i}$, $u_{i+1}$ and $u_{i+2}$ are pairwise distinct for any $i$. Notice that although any transposition $(ab)$ can also be written $(ba)$, the order of the $u_i$'s in the above notation is not indifferent: each $u_i$ is characterized as being the sole element shared by the $i^{\text{th}}$ and the $(i+1)^{\text{th}}$ transposition of the decomposition. 

\medskip

%A permutation $\al$ admits of several alternated factorizations, and we want single out one of them, for each $\al$, in order to allow an inductive reasoning build on this particular decomposition. There is no canonical way to choose one specific alternated factorization. In the following lemma, we roughly describe one possible way to perform this task. 
A permutation $\al$ admits of several alternated factorizations, and we want to single out one of them for each $\al$, in order to allow an inductive reasoning build on this particular decomposition. There is no canonical way to perform this task. In the following lemma, we roughly describe one possible choice, that implicitely refers to the graph $\graphe_d$ on domain $\perm(d)$ whose edges correspond to those pairs of permutations $(\al,\beta)$ such that $\beta=\al(ij)$ for some $i,j\in\{1,\dots,d\}$.

% lemma
\begin{lemma}\label{thm:choix d'une factorisation}
There exists an oriented tree \textdef{$\T_d$} covering $\perm(d)$ which is rooted at $\id$ and such that each $\T_d$-path starting at $\id$, say $\id\al_{1}\dots\al_{p}$, corresponds to an alternated factorization of $\al_{p}$. 
\end{lemma}

% proof
\begin{proof}
Trees $\T_d$ for $d\geq 2$ are defined inductively. For $d=2$, there is a unique such tree: $(12)\imp (21)=(12)(21)$. So, assume we are given $\T_{d-1}$ and carry on the construction of $\T_{d}$ as follows: 
\begin{enumerate}[$(a)$]
\item  First, view each permutation $\al=\al_{1}\dots\al_{d-1}\in\perm(d-1)$ as a permutation of $\{2,\dots,d\}$ by renaming $\al_{i}$ as $\al_{i}+1$. That is, replace $\al=\al_{1}\dots\al_{d-1}$ by $\al^+=(\al_{1}+1)\dots(\al_{d-1}+1)$. 
\item  Then, replace each such $\al^+$ by $\beta=\beta_{1}\dots\beta_{d}\in\perm(d)$ defined by: $\beta_{1}=1$ and $\beta_{i}=[\al^+]_{i-1}=\al_{i-1}+1$ for $i>1$. Thus, $\T_{d-1}$ now covers the set of permutations  $\beta\in\perm(d)$ that fulfill $[\beta]_{1}=1$. 
\item  For each node $\beta$ thus obtained, create a new node labelled by the composition of $\beta$ with the transposition $(1d)$ -- that is by the permutation $\beta(1d)$ -- and create an edge $\beta\dte\beta(1d)$. 
\item  Finally, link each such node $\beta(1d)$ to $d-2$ new nodes $\beta(1d)(di)$, for $i=2,\dots,d-1$.
\end{enumerate}
In Fig~\ref{fig:bonnefact}, we display the steps of the construction of $\T_4$ from $\T_2$. %factorizations of $\perm(4)$ obtained with this method, starting with a covering tree of $\perm(2)$. 
Letters $(a)$, \dots, $(d)$ in the figure refer to the above items.  
The correction of the method on this example is clear. We leave it to the reader to verify that it generalises to any $d$. 
% figure
\begin{figure}[h]
\begin{center}
\includegraphics[height=7cm]{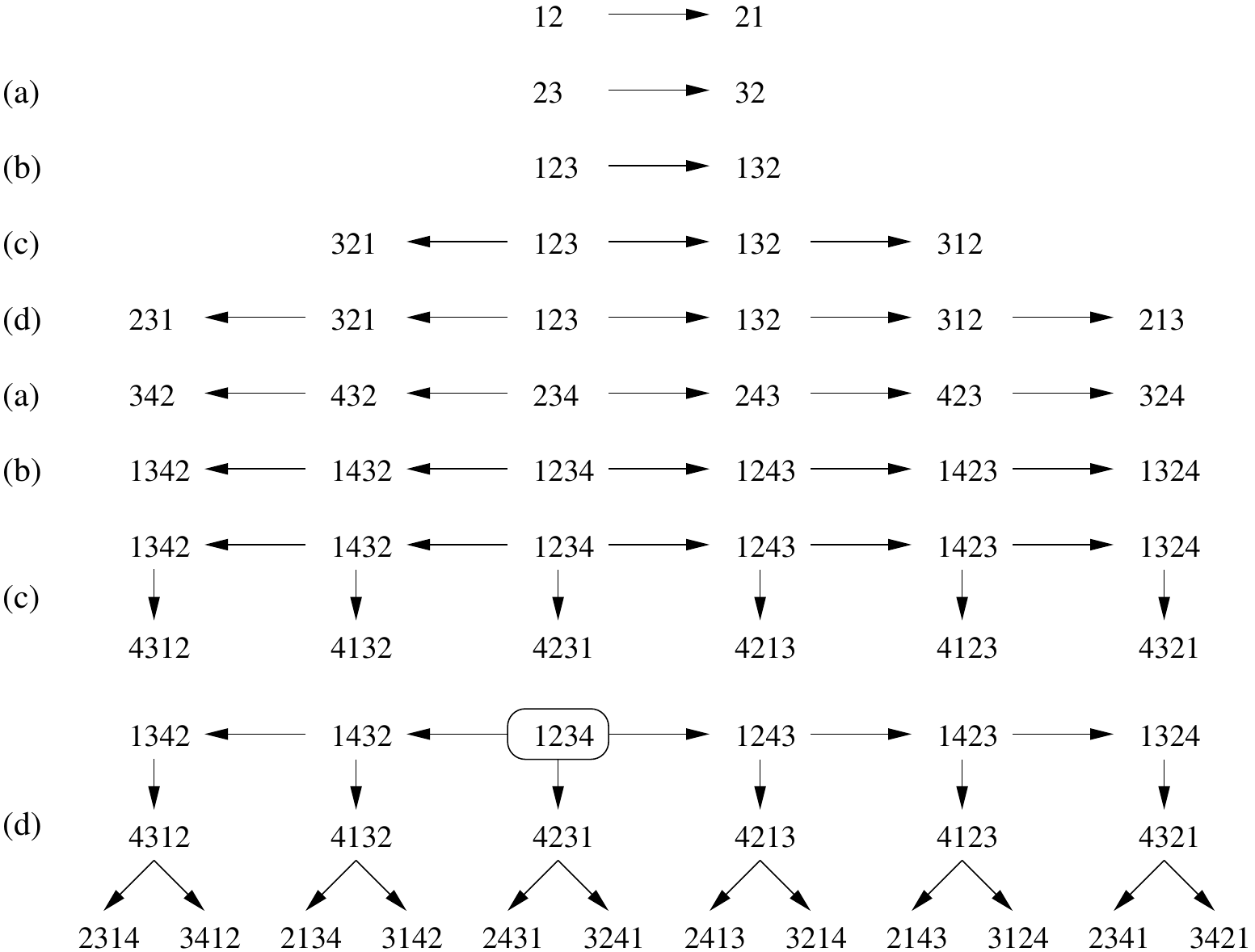}
\caption{\small Construction of $\T_{4}$ from $\T_{2}$. The result is an oriented tree rooted at $\id$, spanning $\perm(4)$, whose all pathes from $\id$ are alternated.}
\label{fig:bonnefact}
\end{center}
\end{figure}
\end{proof}

This lemma allows us to choose, for each $\al\in\perm(d)$, \emph{one} alternated factorization of $\al$: it is the decomposition $(di_1)(i_1i_2)\dots (i_{k-1}i_{k})$ corresponding to the unique path from $\id$ to $\al$ in $\T_d$. We denote by \textdef{$\id.di_1.i_1i_2.\,\dots\,.i_{k-1}i_{k}$} this particular factorization. And when this path until $\al$ can be continued \emph{in $\T_d$} to some permutation $\al(i_ki_{k+1})$, we denote by \textdef{$\al.i_ki_{k+1}$} this last permutation. 
For instance, in the example displayed in Fig~\ref{fig:bonnefact}, we can write $2143=4123.13$ and $3124=4123.14$ while $4321=4123(24)$ \emph{cannot be written} $4123.24$. 
Notice furthermore that the integers $i_k$ and $i_{k+1}$ are ordered in the notation $\al.i_ki_{k+1}$ (unlike in the notation $\al(i_ki_{k+1})$): we place in first position the integer $i_{k}$ involved in the last  transposition leading to $\al$ (with $i_{k}=d$ if $\al=\id$). 
All in all, the reader is invited to keep in mind the numerous presuppositions attached to the notation $\al.uv$: the statement $\beta=\al.uv$ means: 
\begin{itemize}
\item  $\beta=\al(uv)$ ; 
\item  $(\al,\beta)$ is an edge of $\T_d$ ; 
\item  either $\al=\id$ and $u=d$, or $\al=\gamma.tu$ for some $\gamma\in\perm(d)$ and some $t\neq u$ in $\{1,\dots,d\}$. % distinct from $u$. 
\end{itemize}
%Moreover, we chose to order elements $u$ and $v$ in the writing  $\al.uv$. This is done as follows: if $\al=\id$, then necessarily either $u$ or $v$ is $d$ ; we set $u=d$. Otherwise, $\al$ is itself of the form $\beta.xy$ and, by alternation, either $u$ or $v$ is $y$ ; we set $u=y$. Finally, a path in $\T_d$ has the form: $$\id\imp\id.du_1\imp\id.du_1.u_1u_2\imp\dots\imp\id.du_1.u_1u_2\dots u_{k-2}u_{k-1}.u_{k-1}u_{k}$$

%Furthermore, we inductively define $\textdef{$\al.u_1\dots u_{k}$}\egaldef(\al.u_1\dots u_{k-1}).u_{k-1}u_{k}$. That is, $\al.u_1\dots u_{k}$ stands for $\al(u_1u_2)(u_2u_3)\dots (u_{k-1}u_{k})$ when the sequence $$(u_1u_2)(u_2u_3)\dots (u_{k-1}u_{k})$$ labels a path starting at $\al$ in $\T_d$. 

\medskip
Let us now come to a straightforward remark connecting $d$-tuples to alternated factorizations. First recall that for $\tu x=(x_1,\dots,x_d)$, $i\in\{1,\dots,d\}$ and $\al,\beta\in\perm(d)$, we denoted by $[\tu x]_{i}$ the $i^{\text{th}}$ component of $\tu x$, we defined $\tu x_{\al}$ as the $d$-tuple $(x_{\al_1},\dots,x_{\al_d})$, and we noticed that $(\tu x_{\al})_{\beta}=\tu x_{\al\beta}$. (See Definition~\ref{def:tuples et permutations}.) 

% fact
\begin{fact}\label{thm:a propos de x.ij}
For any $\tu x\in\dom{n}^d$, $\al\in\perm(d)$, and $i,j\in\{1,\dots,d\}$: 
\begin{enumerate}[$(a)$]
\item\label{item:x.i...j.jk=x.i...k}  $\tu x_{\al.ij}=(\tu x_{\al})_{(ij)}$.
\item\label{item:(x.i...j)_j=(x)_i}  $[\tu x_{\al.ij}]_{j}=[\tu x]_{d}$. 
\end{enumerate}
\end{fact}

% proof
\begin{proof}
\refitem{item:x.i...j.jk=x.i...k} From $\al.ij=\al(ij)$ and $\tu x_{\al\beta}=(\tu x_\al)_\beta$. 
\refitem{item:(x.i...j)_j=(x)_i} By induction on $\al$:  if $\al=\id$, then necessarily $i=d$ and $[\tu x_{\id.dj}]_{j}=[\tu x]_{d}$ clearly holds. Otherwise, $[\tu x_{\al.ij}]_{j}=[(\tu x_{\al})_{ij}]_{j}=[\tu x_{\al}]_{i}=[\tu x]_{d}$, by induction hypothesis. 
\end{proof}

Given a $d$-tuple $\tu x=(x_1,\dots,x_d)$ of first-order variables, we denote by $\tu x^-$ the $(d-1)$-tuple obtained from $\tu x$ by erasing its last component. That is, $$(x_1,\dots,x_{d-1},x_d)^-\egaldef(x_1,\dots,x_{d-1}).$$ In particular, for $\al\in\perm(d)$, we denote by $\tu x_{\al}^-$ the $(d-1)$-tuple $(x_{\al_1},\dots,x_{\al_{d-1}})$. Each $(d-1)$-tuple build upon the $d$ variables $x_1,\dots,x_d$ can clearly be written $\tu x_{\al}^-$ for some $\al\in\perm(d)$. 
Therefore, the occurence of a non-sorted atom in $\Phi$ has the form $Q(\tu x_{\al}^-)$ for some $Q\in (Q_a)_{a\in\Sg}$ and some $\al\in\perm(d)$, and the purpose of this section amounts to rewrite each such occurence $Q(\tu x_{\al}^-)$ as $Q'(\tu x)$ for some well chosen relation $Q'$. 

% definition
\begin{definition}\label{def:simulation}
Given a $(d-1)$-ary relation $Q$ and two family of $d$-ary relations, $(T_{\al})_{\al\in\perm(d)}$ and $(Q_{\al})_{\al\in\perm(d)}$, we say that $(T_{\al},Q_{\al})_{\al\in\perm(d)}$ is a \textdef{$d$-simulation of $Q$} if the following axioms hold, for any $\al\in\perm(d)$ and any $i,j\leq d$ such that $\al.ij$ is defined, and for any $d$-tuple $\tu x$ of variables: 
\begin{enumerate}[$(\axiom_1)$]
\item\label{item:Tid=Q}  $T_{\id}(\tu x)=Q(\tu x^-)$ if $[\tu x]_{d}=0$.
\item\label{item:Tij(x)=Tij(xij)}  $T_{\al.ij}(\tu x)=T_{\al.ij}(\tu x_{(ij)})$.
\item\label{item:T(x)=Tij(x)}  $T_{\al}(\tu x)=T_{\al.ij}(\tu x)$ if $[\tu x]_{i}=0$. 
\item\label{item:Qid=Tid}   $Q_{\id}(\tu x) = Q(\tu x^-)$.
\item\label{item:Ta(x)=Qa(x)}   $Q_{\al.ij}(\tu x) = T_{\al.ij}(\tu x)$ if $[\tu x]_j=0$.
\item\label{item:Qa indep de xi}  $Q_{\al.ij}(\tu x)$ doesn't depend on $[\tu x]_j$. 
\end{enumerate}
\end{definition}

% lemma
\begin{lemma}\label{thm:Qa(x)=Q(xa)}
Let $(T_{\al},Q_{\al})_{\al\in\perm(d)}$ be a $d$-simulation of some $(d-1)$-ary predicate $Q$. For any $\tu x\in\dom n^d$ and $\al\in\perm(d)$: \, $Q_{\al}(\tu x_{\al})=Q(\tu x^-).$ 
\end{lemma}

% proof
\begin{proof}
Let us first prove that for any $\tu x\in\dom n^d$ and $\al\in\perm(d)$: 
\begin{equation}\label{thm:T(x)=Ta(xa)} 
[\tu x]_d=0\Imp T_{\al}(\tu x_{\al})=Q(\tu x^-). 
\end{equation}

We procced by induction on $\al$. If $\al=\id$, \refeq{thm:T(x)=Ta(xa)} follows from~\refaxiom{item:Tid=Q}. Given a non-identique permutation $\al.ij$, we have: 
\begin{eqnarray*}
T_{\al.ij}(\tu x_{\al.ij}) 
 & = & T_{\al.ij}((\tu x_{\al})_{(ij)}) \ \text{ by Fact~\ref{thm:a propos de x.ij}-(\ref{item:x.i...j.jk=x.i...k}). } \\
 & = & T_{\al.ij}(\tu x_{\al})  \ \text{ by \refaxiom{item:Tij(x)=Tij(xij)}. } 
\end{eqnarray*}
But $[\tu x_{\al}]_{i}=[(\tu x_{\al})_{(ij)}]_{j}=[\tu x_{\al.ij}]_{j}$ and hence, by Fact~\ref{thm:a propos de x.ij}-(\ref{item:(x.i...j)_j=(x)_i}): $[\tu x_{\al}]_{i}=[\tu x]_{d}=0$. Therefore we can resume the above sequence of equalities with: 
\begin{eqnarray*}
T_{\al.ij}(\tu x_{\al.ij}) 
 & = & T_{\al}(\tu x_{\al})  \ \text{ by \refaxiom{item:T(x)=Tij(x)} since $[\tu x_{\al}]_{i}=0$. } \\
 & = & Q(\tu x^-)  \ \text{ by induction hypothesis. }
\end{eqnarray*}
This completes the proof of \refeq{thm:T(x)=Ta(xa)}

\medskip

Let us now prove the equality $Q_{\al}(\tu x_{\al})=Q(\tu x^-)$. If $\al=\id$, the result comes from~\refaxiom{item:Qid=Tid}. 
For a non-identique permutation $\al.ij$, we have to prove $Q_{\al.ij}(\tu x_{\al.ij})=Q(\tu x^-)$ for any tuple $\tu x\in\dom n^d$. First notice that we can restrict, without loss of generality, to the case where $[\tu x]_d=0$. Indeed, we can otherwise consider the tuple $\tu y$ obtained from $\tu x$ by setting $[\tu x]_{d}$ to $0$. (That is, $\tu y$ only differs from $\tu x$ by its $d^{th}$ component, which is null.) 
Clearly, $\tu y^-=\tu x^-$. 
Besides, the $j^{th}$ component of $\tu x_{\al.ij}$ is $[\tu x]_d$, from Fact~\ref{thm:a propos de x.ij}-(\ref{item:(x.i...j)_j=(x)_i}). Similarly, the $j^{th}$ component of $\tu y_{\al.ij}$ is $[\tu y]_d$. Hence, the tuples $\tu x_{\al.ij}$ and $\tu y_{\al.ij}$ coincide on each component of rank distinct from $j$. Therefore $Q_{\al.ij}(\tu x_{\al.ij})=Q_{\al.ij}(\tu y_{\al.ij})$ by~\refaxiom{item:Qa indep de xi} and we get: $Q_{\al.ij}(\tu x_{\al.ij})=Q(\tu x^-)$ iff $Q_{\al.ij}(\tu y_{\al.ij})=Q(\tu y^-)$. 
Thus, we can assume $[\tu x]_d=0$. It follows $[\tu x_{\al.ij}]_{j}=[\tu x]_{d}=0$, by Fact~\ref{thm:a propos de x.ij}-(\ref{item:(x.i...j)_j=(x)_i}), and hence:
\begin{eqnarray*}
Q_{\al.ij}(\tu x_{\al.ij}) 
& = & T_{\al.ij}(\tu x_{\al.ij})  \text{ by \refaxiom{item:Ta(x)=Qa(x)}, since $[\tu x_{\al.ij}]_{j}=0$. } \\
& = & Q(\tu x^{-}) \text{ by ~\refeq{thm:T(x)=Ta(xa)}, since $[\tu x]_d=0$. } 
\end{eqnarray*}
The proof is complete. 
\end{proof}

% lemma
\begin{lemma}
Let $Q$ be a $(d-1)$-ary relation and $(T_{\al})_{\al\in\perm(d)}$, $(Q_{\al})_{\al\in\perm(d)}$ be two tuple of $d$-ary relations satisfying, for each $d$-tuple $\tu x=(x_1,\dots,x_d)$:
\begin{enumerate}[$(\axiomlog_1)$]
\item\label{item:log:Tid=Q}  $\min(x_d)\imp \left(T_{\id}(\tu x)\ssi Q(\tu x^-)\right)$.
\item\label{item:log:Tij(x)=Tij(xij)}  $\left(\neg\max(x_i)\et\neg\max(x_j)\right)\imp \left(T_{\al.ij}(\tu x^{(i)})\ssi T_{\al.ij}(\tu x^{(j)})\right)$.
\item\label{item:log:T(x)=Tij(x)}  $\min(x_i)\imp \left(T_{\al}(\tu x)\ssi T_{\al.ij}(\tu x)\right)$.
\item\label{item:log:Qid=Tid}   $Q_{\id}(\tu x) \ssi Q(\tu x^-)$.
\item\label{item:log:Ta(x)=Qa(x)}  $\min(x_j)\imp \left(Q_{\al.ij}(\tu x)\ssi T_{\al.ij}(\tu x)\right)$.
\item\label{item:log:Qa indep de xi}  $Q_{\al.ij}(\tu x)\ssi Q_{\al.ij}(\tu x^{(j)})$. 
\end{enumerate}
Then $(T_{\al},Q_{\al})_{\al\in\perm(d)}$ is a $d$-simulation of $Q$. Furthermore, each $Q$ admits such a $d$-simulation fulfilling \refaxiomlog{item:log:Tid=Q}\dots\refaxiomlog{item:log:Qa indep de xi}
\end{lemma}

% proof
\begin{proof}
Clearly, the formula~$(\axiomlog_{i})$ is a mere transcription of the axiom~$(\axiom_{i})$ for each $i\neq 2$. We have to prove that \refaxiomlog{item:log:Tij(x)=Tij(xij)} implies \refaxiom{item:Tij(x)=Tij(xij)}. 
Formula~\refaxiomlog{item:log:Tij(x)=Tij(xij)} yields that $T_{\al.ij}(\tu x)$ has the same value on tuples of the form 
\[\tu x = (\tu u,x+c,\tu v,y-c,\tu w),\] 
for any $c\leq\min\{n-1-x,y\}$, where $x+c$ (resp. $y-c$) is the component of rank $i$ (resp. $j$) of $\tu x$. (That is: $\tu u\in\dom n^{i-1}$, $\tu v\in\dom n^{j-i-1}$ and $\tu w\in\dom n^{d-j}$.) 
In other words, the value of $T_{\al.ij}(\tu x)$ depends on $[\tu x]_i+[\tu x]_j$ rather than on the precise values of these two components. 
As a consequence, for any $\tu u\in\dom n^{i-1}$, $\tu v\in\dom n^{j-i-1}$, $\tu w\in\dom n^{d-j}$ and any $x,y\in\dom n$: 
\[T_{\al.ij}(\tu u,x,\tu v,y,\tu w)=T_{\al.ij}(\tu u,y,\tu v,x,\tu w).\] This is Axiom~\refaxiom{item:Tij(x)=Tij(xij)}.

\medskip

It remains to prove that such relations $(T_{\al})_{\al\in\perm(d)}$ and $(Q_{\al})_{\al\in\perm(d)}$ exist for every $Q$. To see this, assume we are given a $(d-1)$-ary relation $Q$ and define the $Q_{\al}$'s and the $T_{\al}$'s as follows: 
\begin{itemize}
\item  $Q_{\al}(\tu x)=Q(\tu x_{\al^{-1}}^-)$ for any $\tu x\in\dom n^{d}$ ; 
\item  $T_{\id}(\tu x)=Q_{id}(\tu x)$ for any $\tu x\in\dom n^{d}$ ; 
\item  $T_{\al.ij}(\tu u,x,\tu v,y,\tu w)=Q_{\al.ij}(\tu u,x+y,\tu v,0,\tu w)$ for any $x,y\in\dom n$, $\tu u\in\dom n^{i-1}$, $\tu v\in\dom n^{j-i-1}$ and $\tu w\in\dom n^{d-j}$.
\end{itemize}
We leave it to the reader to check that the sequence $(T_{\al},Q_{\al})_{\al\in\perm(d)}$ satisfies the formulas \refaxiomlog{item:log:Tid=Q}, \dots, \refaxiomlog{item:log:Qa indep de xi} (and hence, is a $d$-simulation of $Q$). 
\end{proof}

% proposition
\begin{proposition}\label{thm:normal-sort-Q} 
For any $d>0$, $\esorel{}(\fa^{d},\arity d,\semisorted)\subseteq\esorel{}(\fa^{d},\arity d,\sorted)$ on $\coordstruc{d-1}$.
\end{proposition}

\medskip

Propositions~\ref{thm:normal-sort-R} and~\ref{thm:normal-sort-Q} can now be collected in te following result: 

% theorem
\begin{theorem}\label{thm:esorel(d,d)=esorel(d,d,sorted)} 
$\eso{}(\fa^{d},\arity d)=\eso{}(\fa^{d},\arity d,\sorted)$ on $\coordstruc{d-1}$ for any $d>1$. 
\end{theorem}

\end{document}